\documentclass[aps,prx,twocolumn,superscriptaddress,longbibliography]{revtex4-2}

\usepackage{microtype}
\usepackage{graphicx}
\usepackage{booktabs} 

\usepackage{hyperref}

\usepackage{bm}
\usepackage{physics}
\usepackage{dsfont}
\usepackage[export]{adjustbox} 

\usepackage{xcolor}
\definecolor{customPink}{cmyk}{0, 0.5, 0.2, 0}

\usepackage{amsthm}

\usepackage[capitalize]{cleveref}

\newcommand{\bthe}{\bm{\theta}}
\newcommand{\be}{\mathbf{e}}
\newcommand{\cE}{\mathcal{E}}
\newcommand{\ot}{\otimes}

\newcommand{\id}{\mathds{1}}
\newcommand{\fl}{\mathds{F}}

\newcommand{\Ket}[1]{\left.\left| #1 \right\rangle\!\right\rangle}

\newcommand{\Braket}[1]{\left\langle\!\left\langle #1 \right\rangle\!\right\rangle}

\newcommand{\dg}{\dagger}

\newcommand{\cL}{\mathcal{L}}
\newcommand{\cH}{\mathcal{H}}

\newcommand{\bE}{\mathop{\mathbf{E}}}
\newcommand{\supp}{\mathrm{supp}}
\newcommand{\cP}{\mathcal{P}}
\newcommand{\lcs}{\Gamma}

\newcommand{\wg}{\mathrm{Wg}}
\renewcommand{\Tr}{\mathrm{Tr}} 
\newcommand{\Var}{\mathrm{Var}}
\newcommand{\dD}{\mathds{D}}
\newcommand{\lone}[1]{\left\Vert #1 \right\Vert_{1}}
\newcommand{\ltwo}[1]{\left\Vert #1 \right\Vert_{\infty}}
\newcommand{\fbnorm}[1]{\left\Vert #1 \right\Vert_{\rm F}}

\newtheorem{theorem}{Theorem}
\newtheorem{corollary}[theorem]{Corollary}

\newtheorem{remark}{Remark}
\newtheorem{definition}[]{Definition}

\newtheorem{example}{Example}

\Crefname{equation}{Eq.}{Eqs.}
\Crefname{figure}{Fig.}{Figs.}
\Crefname{tabular}{Tab.}{Tabs.}
\Crefname{section}{Sec.}{Secs.}

\begin{document}

\title{Barren Plateaus Beyond Observable Concentration}

\author{Zi-Shen Li}
\affiliation{Quantum Information and Computation Initiative, Department of Computer Science, School of Computing and Data Science, The University of Hong Kong, Pokfulam Road, Hong Kong, China}

\author{Bujiao Wu}
\email{wubujiao@iqasz.cn}
\affiliation{International Quantum Academy, Shenzhen 518048, China}

\author{Xiao-Wei Li}
\affiliation{Department of Physics, College of Physics, Chengdu University of Technology, Chengdu, 610059, China}

\author{Man-Hong Yung}
\affiliation{International Quantum Academy, Shenzhen 518048, China}

\begin{abstract}
Parameterized quantum circuits (PQCs) are central to quantum machine learning and near-term quantum simulation, but their scalability is often hindered by barren plateaus (BPs), where gradients decay exponentially with system size. Prior explanations, including expressivity, entanglement, locality, and noise, are often presented in ways that conflate two distinct issues: concentration of the measured observable and loss of parameter sensitivity caused by circuit dynamics.
We develop a unified statistical framework that separates these mechanisms. We show that several standard BP explanations, including locality- and entanglement-related effects, can be understood through a single phenomenon that we term observable concentration (OC). Importantly, we prove that avoiding OC is necessary but not sufficient for trainability. Beyond OC, we identify two distinct mid-circuit sources of gradient suppression. 
{First, in circuits with effectively independent forward and backward blocks, parameter perturbations can propagate outside the measurement light cone and become inaccessible to the final observable, yielding information-loss-induced BPs.}
{Second, BPs can also arise in circuits with highly correlated forward and backward blocks, as we demonstrate through an echo-type circuit model that is reminiscent of information scrambling.}
\end{abstract}

\maketitle

\section*{Introduction}

Quantum computing faces substantial experimental and theoretical challenges. While certain quantum algorithms, such as Shor's algorithm~\cite{shorPolynomialTimeAlgorithmsPrimeFactorizationDiscreteLogarithms1997,monz2016RealizationScalableShorAlgorithm}, Grover's search~\cite{groverFastQuantumMechanicalAlgorithmDatabaseSearch1996}, and the HHL algorithm~\cite{harrowQuantumAlgorithmLinearSystemsEquations2009}, offer provable quantum advantages, their practical realization on near-term quantum devices remains difficult. This experiment-theory gap primarily arises from the challenge of designing quantum algorithms that can simultaneously ensure provable advantages while remaining implementable on current hardware.

Parameterized quantum circuits (PQCs) have emerged as a versatile framework for demonstrating the capabilities of quantum devices, with applications spanning quantum machine learning~\cite{biamonte2017QuantumMachineLearning,schuldQuantumMachineLearningFeatureHilbertSpaces2019,meyerExploitingSymmetryVariationalQuantumMachineLearning2023}, variational quantum eigensolvers~\cite{peruzzo2014VariationalEigenvalueSolverPhotonicQuantumProcessor,kandalaHardwareefficientVariationalQuantumEigensolverSmallMolecules2017,carolanVariationalQuantumUnsamplingQuantumPhotonicProcessor2020}, and quantum simulations~\cite{li2017EfficientVariationalQuantumSimulatorIncorporatingActive,kokail2019SelfverifyingVariationalQuantumSimulationLatticeModels,endoVariationalQuantumSimulationGeneralProcesses2020}. 
A central theoretical challenge in deploying PQCs for these tasks is the phenomenon of barren plateaus (BPs).
Barren plateaus manifest as exponential flatness in the training landscape, arising from various factors including entanglement \cite{ortizmarreroEntanglementInducedBarrenPlateaus2021,patti2021EntanglementDevisedBarrenPlateauMitigation}, cost function locality \cite{cerezoCostFunctionDependentBarrenPlateausShallow2021,uvarov2021BarrenPlateausCostFunctionLocalityVariational}, the presence of noise \cite{wangNoiseinducedBarrenPlateausVariationalQuantumAlgorithms2021}, etc.
As a major obstacle to the scalability of PQCs, BPs impose an exponential overhead in the computational cost of gradient estimation.
In addition, recent research has proposed several approaches to mitigate or avoid BPs, including the design of BP-free circuit architectures~\cite{pesahAbsenceBarrenPlateausQuantumConvolutionalNeural2021}, the development of diagnostic tools for BP detection~\cite{patti2021EntanglementDevisedBarrenPlateauMitigation,laroccaDiagnosingBarrenPlateausToolsQuantumOptimal2022}, the formulation of training strategies to mitigate BPs~\cite{sackAvoidingBarrenPlateausUsingClassicalShadows2022}, etc.

\begin{figure}[t]
    \centering
    \includegraphics[width=1\linewidth]{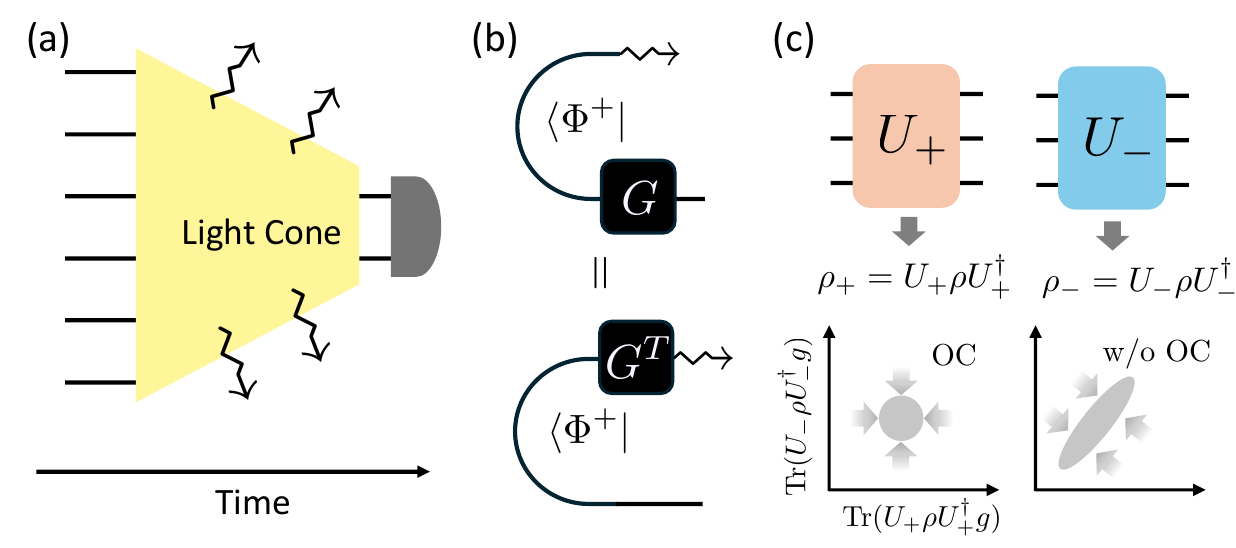}
    \caption{Demonstration of information loss and gradient vanishing. (a) depicts information loss in variational quantum circuits. (b) shows a simple illustration on how this information loss results in vanishing response. (c) shows different statistical behaviors of BPs with and without observable concentration (OC).}
    \label{fig:lightcone-concept}
\end{figure}

Although these approaches offer a pathway to circumvent barren plateaus in PQCs, they may inadvertently imply classical simulability. This reveals a critical dilemma: mitigating barren plateaus may, in some cases, eliminate the very quantum advantage that PQCs are intended to achieve \cite{cerezo2025DoesProvableAbsenceBarrenPlateausImply}. Moreover, the presence of gate noise~\cite{shaoSimulatingNoisyVariationalQuantumAlgorithmsPolynomial2024} or certain types of circuit randomness~\cite{angrisaniClassicallyEstimatingObservablesNoiselessQuantumCircuits} can also render PQCs classically simulable. In such cases, the loss function estimation becomes concentrated within a classically accessible regime, making the computation either deterministically~\cite{shaoSimulatingNoisyVariationalQuantumAlgorithmsPolynomial2024,martinez2025EfficientSimulationParametrizedQuantumCircuitsNonunitala} or probabilistically~\cite{angrisaniClassicallyEstimatingObservablesNoiselessQuantumCircuits} simulable. These observations naturally raise an important question: \textit{Can a parameterized quantum circuit simultaneously exhibit quantum advantage while remaining free of barren plateaus?} To identify and characterize such regimes, it is essential to develop a deeper understanding of the relationship between BPs and circuit architecture, necessitating a comprehensive reassessment of BPs.


{A recurring ambiguity in the literature is that BP is sometimes associated not only with gradient concentration, but also with concentration of the cost (observable) values themselves.}
{These notions coincide in many commonly studied settings, but they are not equivalent: by the parameter-shift rule, a gradient is a difference of two cost evaluations, and can therefore vanish either because each evaluation concentrates, or because the two evaluations become nearly indistinguishable due to strong correlations.}
We here define the \emph{observable concentration} (OC) as the phenomenon where the variance of the observable is concentrated exponentially toward its mean. 
In addition, we remark that Angrisani \textit{et al.} \cite{angrisaniClassicallyEstimatingObservablesNoiselessQuantumCircuits} demonstrated that local observables in certain noiseless quantum circuits can be classically estimated within acceptable mean errors. We found that this classical simulability can arise directly from OC. 
In this study, we explore the fundamental origins of BPs, including previously identified factors such as cost function locality \cite{cerezoCostFunctionDependentBarrenPlateausShallow2021} and entanglement \cite{ortizmarreroEntanglementInducedBarrenPlateaus2021}, which we demonstrate are intrinsically linked to OC.

{Our goal is to complement these end-of-circuit variance viewpoints by making explicit an additional, operationally measurable ingredient: the parameter-shift correlation between the two shifted evaluations that define the gradient.}
{This term is shaped by mid-circuit information dynamics and can dominate the onset of BPs even when end-of-circuit observable statistics remain non-concentrated.}
{This also clarifies the relation to recent unified theories of barren plateaus.}
{Lie-algebraic analyses characterize how the ansatz explores accessible operator directions and are particularly relevant to the behavior of deep PQC families~\cite{ragone2024LieAlgebraicTheoryBarrenPlateausDeep,ragoneUnifiedTheoryBarrenPlateausDeepParametrized2023}, while purity- and entanglement-based analyses connect trainability to concentration of reduced states~\cite{ortizmarreroEntanglementInducedBarrenPlateaus2021,patti2021EntanglementDevisedBarrenPlateauMitigation}.}
{Our decomposition does not replace these viewpoints.}
{Instead, we focus on locally scrambling ensembles \cite{caro2023OutofdistributionGeneralizationLearningQuantumDynamicsa,angrisaniClassicallyEstimatingObservablesNoiselessQuantumCircuits} that provide a complementary regime of interest. 
In addition, we separate the end-of-circuit observable fluctuation from the mid-circuit transmission of parameter sensitivity, so that one can diagnose cases where the observable remains non-concentrated but the gradient vanishes to zero.}
{Within this framework, we organize the mid-circuit mechanisms into two circuit settings: independent forward--backward layers, where the dominant mechanism is information loss, and correlated forward--backward layers, where an echo-type construction gives a complementary information-scrambling example.}

{In the independent-layer setting, information loss is illustrated by the toy setting in \Cref{fig:lightcone-concept}. If one share of a maximally entangled pair $\Phi^+$ leaves the causal light cone of the measured region, then the perturbation induced by a local gate $G$ is transferred into degrees of freedom that are inaccessible to the final measurement, leading to an approximately vanishing linear response. Beyond this picture, we derive rigorous upper bounds on the gradient variance under minimal randomness assumptions. In contrast to analyses based on global Haar randomness~\cite{mele2024IntroductionHaarMeasureToolsQuantumInformation}, our results require only local scrambling conditions (for example, a local $2$-design), which better match realistic PQC architectures.}

{In the correlated-layer setting, we analyze an echo-type circuit in which the backward block is the inverse of the forward block. This complementary example shows that information-scrambling-induced suppression of the correlation factor can yield BPs even when the observable statistics remain non-concentrated.}
{We complement the theory with explicit constructions and numerical evidence. In particular, we study a QCNN-inspired hierarchical circuit with linear depth and observe information-loss-induced barren plateaus in regimes where observable concentration is absent.}

\section*{Background}

BPs reflect the trainability of parametrized quantum circuits (PQCs) \cite{mccleanBarrenPlateausQuantumNeuralNetworkTraining2018}, which are usually considered at the ensemble level. A PQC $U(\bthe)$ is parameterized by a set of real-valued parameters $\bthe=(\theta_1,\theta_2,\dots, \theta_m)^T$. In the ensemble approach, we associate each unitary with a probability, i.e., $\cE=\{U,p(U)\}$ and the expectation operation over this ensemble is defined as
$\mathbf{E}_{U\sim\cE}~(\cdot) := \int dU p(U) (\cdot)$,
which captures the typical behavior of the entire set of unitaries.
Any specific parameterization configuration of a PQC defines an ansatz for minimizing a cost function of the form $C(\bthe):=C[U(\bthe)] :=\Tr[U(\bthe)\rho_0 U(\bthe)^\dg H]$ where $\rho_0$ denotes the input state, and $H\in \cL(\cH)$ is a Hermitian observable acting on the Hilbert space $\cH$. 
Here, $\cL(\cH)$ denotes the space of linear operators acting on $\cH$.
The flatness of the training landscape is quantified by the variance of the gradient. 

\begin{definition}[Barren plateaus]
A barren plateau (BP) arises when the gradient of the cost function, given by $\nabla_{\bthe} C[U(\bthe)] = (\partial_{\theta_1}C,\cdots,\partial_{\theta_m}C)^T$ vanishes exponentially with the number of qubits $n$ over the ensemble $\cE$.
Specifically, for some constant $\alpha>1$ and $j\in [m]$,
    \begin{align*}
    \Var_{U\sim\cE}~\partial_{\theta_j} C[U(\bthe)] \in \mathcal{O}(\alpha^{-n}).
    \end{align*} 
\end{definition}
This definition is an ensemble statement about typical gradients under random parameter sampling (e.g., random initialization), and it should be distinguished from the trivial fact that any landscape can contain isolated critical points with exactly zero gradient.

Although gradient-free optimization methods have been proposed \cite{arrasmith2021effect}, they encounter fundamentally equivalent challenges, as the presence of BP implies an exponentially low probability of determining the descent direction in parameter space. Therefore, without loss of generality, we focus our analysis on gradient-based approaches in this work.

\section*{Setup}

We consider a system of $n$ qubits, for which the associated Hilbert space $\mathcal{H}$ has dimension $d = 2^n$. 
We begin by considering $H = c g$, where $g$ is a local Pauli operator and $c \in \mathbf{R}$ is a constant. Extensions to general Hamiltonians will be discussed subsequently. We assume that the ensemble $\cE_{\rm LS}=\{U(\bthe),p(\bthe)\}$ is \emph{locally scrambling} \cite{caro2023OutofdistributionGeneralizationLearningQuantumDynamicsa}, a property satisfied by many circuits of physical interest \cite{angrisaniClassicallyEstimatingObservablesNoiselessQuantumCircuits}. We give the formal definition as follows.

\begin{definition}[Locally scrambling ensemble]\label{def:local-scramb}
    The ensemble $\cE_{\rm LS}$ is said to be \emph{locally scrambling} if for any random unitary $U \sim \cE_{\rm LS}$ and a tensor product of single-qubit Clifford gates $\bigotimes_j V_j$, we have $(\bigotimes_j V_j) U \sim \cE_{\rm LS}$.
\end{definition}

We also let the probability distribution $p(\bthe)$ be translation invariant in $\bthe$, which is equivalent to uniformly sampling $\theta_j$ from $[0,2\pi)$. According to the parameter shift rule \cite{crooksGradientsParameterizedQuantumGatesUsingParametershift2019,liHybridQuantumClassicalApproachQuantumOptimalControl2017}, the derivative of the cost function can be written as
\begin{align*}
    \partial_{\theta_i} C(\bthe) = u\left[C(\bthe+\be_i\pi/4u)-C(\bthe-\be_i\pi/4u)\right],
\end{align*}
where $u$ is a constant that depends on the generator of the parameterized gate and $\be_i$ is the $i$-th unit vector.
Let $\rho(\boldsymbol{\theta}) = U(\boldsymbol{\theta}) \rho_0 U^\dagger(\boldsymbol{\theta})$, and define the shifted states $\rho_\pm(\bm \theta) := \rho(\boldsymbol{\theta} \pm \mathbf{e}_i \pi / 4u)$. 
For simplicity, we use the notation $\rho$ and $\rho_{\pm}$ to denote $\rho(\boldsymbol{\theta})$ and $\rho_\pm(\bm \theta)$, respectively.
The variance can be expressed as
\begin{equation*}
    \Var_{U\sim \cE}~\partial_{\theta_i}C = c^2u^2 \bE_{U\sim \cE}~|\Tr[g\rho_+ - g\rho_-]|^2.
\end{equation*}
By the translation invariance of the ensemble $\mathcal{E}$, this expression simplifies to
\begin{equation}
    \Var_{U\sim \cE}~\partial_{\theta_i}C= 2c^2u^2 (1-r) \bE_{U\sim \cE}~[\Tr(g\rho)]^2.
\label{eq:norm-grad}
\end{equation}
where the ratio 
\begin{equation*}
    r:= \bE_{U\sim \cE}~\Tr[g\rho_+]\Tr[g\rho_-]/\bE_{U\sim \cE}~[\Tr(g \rho)]^2
\end{equation*}
is the \emph{Pearson correlation coefficient}, which quantifies the linear correlation between the measurement outcomes on the two shifted states. In certain cases, $r$ can approach $1$ exponentially with the number of qubits $n$, becoming a significant factor of BPs.
{We interpret $(1-r)$ as a \emph{parameter-shift distinguishability} factor: $r\to 1$ means the two shifted circuits become statistically indistinguishable on the measured observable, leading to vanishing gradients even when $\bE[\Tr(g\rho)]^2$ does not decay exponentially.}
{If $\bE[\Tr(g\rho)]^2$ decays exponentially while $(1-r)$ remains order one, the dominant mechanism is observable concentration.}
{If $\bE[\Tr(g\rho)]^2$ remains non-negligible but $(1-r)$ decays exponentially, the BP is instead caused by loss of parameter sensitivity inside the circuit.}
{Examples~\ref{ex:1} and~\ref{ex:2} illustrate how this correlation factor appears in independent-layer and correlated-layer settings, respectively.}

The quantity $\bE_{U \sim \mathcal{E}} \left[ \Tr(g \rho) \right]^2$ in \cref{eq:norm-grad} characterizes the variance of the observable $g$, given that its expectation vanishes under the locally scrambling ensemble. This term reflects the degree of concentration of measurement outcomes induced by the ensemble at the end of the circuit. While prior work \cite{ragone2024LieAlgebraicTheoryBarrenPlateausDeep} has attributed the vanishment in the variance of observable to BPs, the contribution of the Pearson correlation coefficient—particularly how it is affected by the inner structures of the ansatz—remains largely unexplored.

{We first study when the observable second-moment factor $\bE[\Tr(g\rho)]^2$ becomes exponentially small (observable concentration).}

We begin by demonstrating how BPs can arise from concentration phenomena in end-of-circuit measurement statistics. For a locally scrambling ensemble $\cE$ and a state $\rho = U \rho_0 U^\dg$ with $U \sim \cE$, we derive the following bound on the observable variance of Paulis.

\begin{theorem}[Observable Concentration]\label{thm:obs_conc}
For a locally scrambling ensemble $\cE_{\rm LS}$ and a state $\rho = U \rho_0 U^\dg$ with $U \sim \cE_{\rm LS}$, let $g$ be a Pauli operator with support $A\subseteq [n]$.
The second moment of the observable $g$ is bounded by:
\begin{align}
 \bE_{U\sim \cE_{\rm LS}}~[\Tr(g\rho)]^2 \le \ltwo{g}^2\left(\frac{2}{3}\right)^{|A|} \bE_{U\sim\cE_{\rm LS}}~D_{\rm HS}^2(\rho_A),\label{eq:end-circ-statis}
\end{align}
{where $\ltwo{g}=1$ is the operator norm of the Pauli operator $g$, $\rho_A$ is the reduced density matrix of $\rho$ on the qubits in $A$, and $D_{\rm HS}^2(\rho_A):=\Vert{\rho_A-\id/2^{|A|}}\Vert_{\rm F}^2=\Tr[\rho_A^2]-2^{-|A|}$ is the purity excess over the maximally mixed state.}
\end{theorem}
{A generalization to Hermitian operators, together with the proof of the theorem, is provided in the Methods section.
We remark that \Cref{thm:obs_conc} significantly generalizes standard concentration results: it does not rely on the assumption of global Haar randomness or deep circuits, but requires only local scrambling.
This implies that OC can occur even in shallow or intermediate-depth circuits, provided they satisfy the local scrambling property.}
By \Cref{thm:obs_conc}, observable concentration (OC) can be attributed to two main factors: (1) high entanglement in the end-of-circuit state, which causes the reduced state on subsystem $A$ to approach the maximally mixed state. In the context of quantum neural networks, this effect manifests through the growth of hidden-layer dimensions~\cite{ortizmarreroEntanglementInducedBarrenPlateaus2021}; and (2) the use of global cost functions~\cite{cerezoCostFunctionDependentBarrenPlateausShallow2021,uvarovBarrenPlateausCostFunctionLocalityVariational2021}, which is indicated by the exponent $|A|$ in \Cref{eq:end-circ-statis}. We also note that noise-induced BPs~\cite{wangNoiseinducedBarrenPlateausVariationalQuantumAlgorithms2021} are likely a consequence of OC, which extends beyond the scope of this work. In the following, we investigate how the internal structure of the ansatz contributes to the emergence of BPs, beyond the effects accounted for by OC.
{Thus entanglement-induced concentration, hidden-layer purity effects, and global-cost concentration all enter the same factor in \cref{eq:norm-grad}, whereas the correlation factor captures a distinct loss of parameter sensitivity.}

\section*{BPs in Independent Layers}

\begin{figure}[t]
    \centering
    \includegraphics[width=0.7\linewidth]{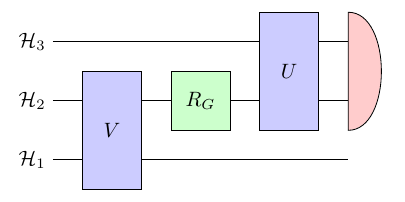}
    \caption{Circuit models to demonstrate mid-circuit response. $\cH_{1},\cH_2$ and $\cH_3$ denote different Hilbert spaces with dimensions $d_1$, $d_2$, and $d_3$ respectively. The measurement is taken within $\cH_2\ot\cH_3$.}
    \label{fig:mid-circ-example1}
\end{figure}

\begin{figure}[t]
    \centering
    \includegraphics[width=0.95\linewidth]{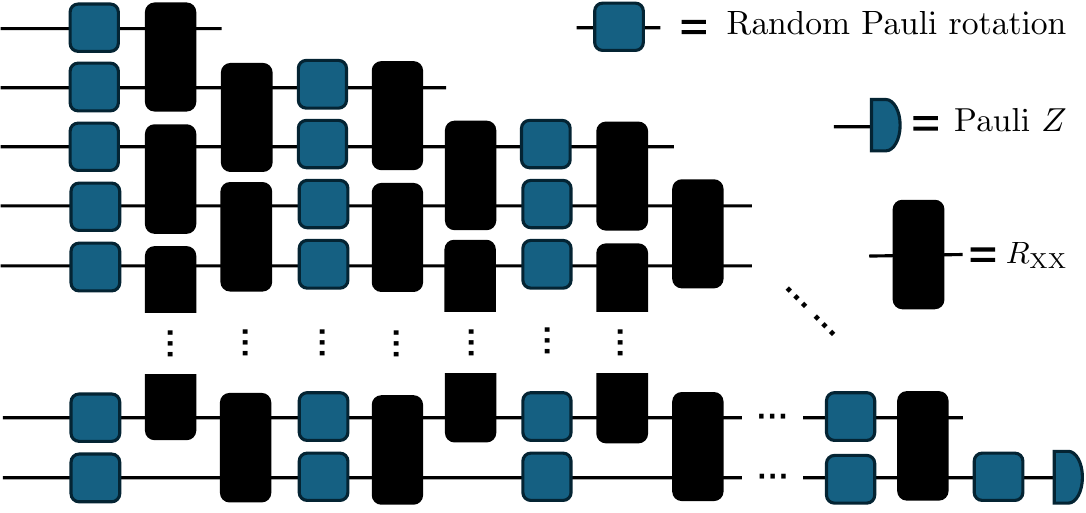}
    \caption{{Hierarchical tree circuit model.}
    This quantum circuit consists of $L$ layers of single-qubit and two-qubit parameterized gates.
    The qubits are indexed sequentially from bottom to top, denoted as $1,2,...,n$.
    Similarly, the layers are labeled consecutively from left to right, labeled as $1,2,...,L$.}
    \label{fig:tree-circ}
\end{figure}

{As the entire process of PQCs involves adaptive updates to the parameters within the circuit, a detailed examination of the mid-circuit behavior within PQCs is essential. To formalize this, we analyze the derivative of a parameterized gate of the form $\exp(-i G \theta)$, where $G$ is a Pauli operator. This gate is positioned between a forward circuit block $U$ and a backward circuit block $V$. In this section, the unitaries $U$ and $V$ are independently drawn from unitary ensembles $\mathcal{E}_U$ and $\mathcal{E}_V$, respectively. For generality, we assume only that $\mathcal{E}_V$ is locally scrambling, while $\mathcal{E}_U$ remains a flexible choice within our model. We begin by considering cost functions of the form $\Tr[O\rho]$, where $\rho$ denotes the state prepared by the circuit, and subsequently extend the analysis to a batch-decomposed formulation.}

To analyze the mid-circuit behavior of the circuit evolution, we start by examining the snapshot of the system state $\rho_V$ immediately following the application of unitary $V$. Let $\Gamma$ denote the light-cone subspace, defined as the set of qubits that back-propagate from the support of the observable through the circuit up to $V$. Specifically, this subspace can be expressed as $\lcs  = \supp(U^\dagger O U)$ with $O$ being an arbitrary observable. The gradient variance is determined by the reduced state of $\tilde{\rho}_V^{\otimes 2}$ on the light-cone region $\Gamma$. After simplification, we find that it admits the following upper bound 
\begin{equation}
  c^2 u^2 \ltwo{O}^2 \lone{\bE_{V\sim \cE_{\rm LS}}~ \Tr_{\neq \lcs}\left(\tilde{\rho}_V^{\ot 2}\right)},
  \label{eq:var_information_loss}
\end{equation}
where $\widetilde{\rho}_V = G \rho_{V} G - \rho_V$ and $\Vert \cdot\Vert_1$ denotes the trace norm. Further details and a full derivation of this result are provided in Methods.

{Before diving into the rigorous derivation, we can gain intuition from \Cref{fig:lightcone-concept} (a) and (b): strong entanglement across the light-cone boundary can transfer the effect of a local perturbation into degrees of freedom outside the measurement light cone.}
{In such cases, the reduced state within the light cone approaches maximally mixed, and the measurement becomes insensitive to the perturbation.}
{Operationally, we call this mechanism information loss because the degrees of freedom carrying the perturbation are excluded by the final measurement light cone.}
{This is distinct from information scrambling in a closed system: in scrambling the perturbation is still present globally, but is delocalized over the full operator space so that a fixed final observable can have exponentially small overlap with it.}
{The fixed observable in this sense need not be local; it may be, for example, a full-support Pauli operator chosen independently of the scrambling unitary.}
{Thus loss is caused by restriction to an accessible subsystem or light-cone region, whereas scrambling is caused by delocalization within the accessible closed-system dynamics.}

Specifically, let $\cP_\Gamma$ denote the set of Pauli basis on $\Gamma$. We define the effective Pauli set as $S(\Gamma,G):= \{\id, g|g\in\cP_\Gamma, gG + Gg = 0\}$ and introduce the effective reduced state $\varrho_{\lcs}(V) := 2^{-|\lcs|} \sum_{g\in S(G,\lcs)} g\Tr[g \rho_{\lcs}(V)]$. The corresponding information loss is then characterized by the Hilbert-Schmidt deviation of $\varrho_{\lcs}(V)$ from the maximally mixed state. 
{We emphasize that while $\varrho_{\lcs}$ is not generally equal to the reduced state $\rho_\Gamma$, it captures the specific information component recoverable by the gradient operator $G$.}
{By projecting onto the subspace of operators anti-commuting with $G$, $\varrho_{\lcs}$ provides a tighter bound than $\rho_\Gamma$ for locally scrambling ensembles, serving as a more precise statistical signature of trainability.}

\begin{theorem}[Information loss and barren plateaus]\label{obs:mid-circ}
Let $O$ be an arbitrary Hermitian operator and its expectation value be the cost function. Let $\lcs$ denote the light cone subspace after applying the locally scrambling unitary $V$.
The variance of the gradient is upper bounded as:
    \begin{align*}
        \Var\left(\partial_{\theta} C \right) \le u^2 \ltwo{O}^2 2^{|\lcs|+2}\bE_{V \sim \cE_{\rm LS}}~D_{\rm HS}^2[\varrho_{\lcs}(V)].
    \end{align*}
\end{theorem}

{Notably, \Cref{obs:mid-circ} does not rely on Haar randomness for either $U$ or $V$; it is sufficient that the circuit $V$ exhibits local scrambling while remaining independent of the distribution of $U$. This condition highlights the irreversible character of information loss, which cannot be remedied by any post-processing unitaries. This also suggests that BPs can emerge even in the absence of OC if we choose $U$ carefully enough to maintain non-concentrated end-circuit statistics.}
{Indeed, information loss provides a concrete mid-circuit mechanism to realize BPs without OC.}
To demonstrate this, we must isolate the factor of OC from the gradient variance, which reveals the key quantity $r$ as shown in \Cref{eq:norm-grad}. We then present several concrete examples to elaborate on this.

Consider the circuit structure shown in Fig.~\ref{fig:mid-circ-example1}, where the unitary $V$ is drawn from a unitary 2-design, the operator $G$ is a Pauli observable with support restricted within the subspace $\mathcal{H}_2$, i.e., $\supp(G) \subset \mathcal{H}_2$.
{We assume the initial state is a product state $\rho_0 = \rho_{12} \otimes \rho_3$, and $U$ acts trivially on the partition boundary such that separability is maintained up to $V$.}
Due to the circuit geometry, only measurements performed on $\cH_2$ and $\cH_3$ can detect changes in the parameter $\theta$. Focusing on the snapshot of the system after applying the unitary $V$, the light cone subspace is $\cH_2 \ot \cH_3$ and the state can be written as $\rho_V = \rho_{V}^{(1,2)}\ot \rho_V^{(3)}$, where the superscripts indicate different subspaces. According to \Cref{obs:mid-circ}, we obtain the effective reduced state as $\varrho_{\lcs}=2^{-|\lcs|} \sum_{g_{i}\ot g_{j}\in\cP(G,\lcs)} g_i^{(2)}\ot g_j^{(3)} \Tr\left[g_i^{(2)} \rho^{(2)}_V\right]\Tr\left[g_i^{(3)} \rho^{(3)}_V\right]$. As the state is of product form between $\cH_2$ and $\cH_3$, we can write the deviation in a simpler form and disregard the subsystem $\cH_3$ through the inequality, i.e., $D_{\rm HS}^2(\varrho_{\lcs}) = D_{\rm HS}^2\left[\rho_{V}^{(2)}\right] \Tr[\rho^{(3)2}]\le D_{\rm HS}^2\left[\rho_{V}^{(2)}\right]$. Now the result becomes clear: the variance of gradient is connected to the Hilbert-Schmidt deviation of the reduced state in $\cH_2$ from the maximally mixed state, which can be easily achieved for any $V$ that is random enough. The construction in the following example, obtained by taking $U$ and $V$ from a unitary 2-design, serves as a special case of Theorem~\ref{obs:mid-circ}.

\begin{figure*}[t]
    \centering
    \includegraphics[width=0.95\linewidth]{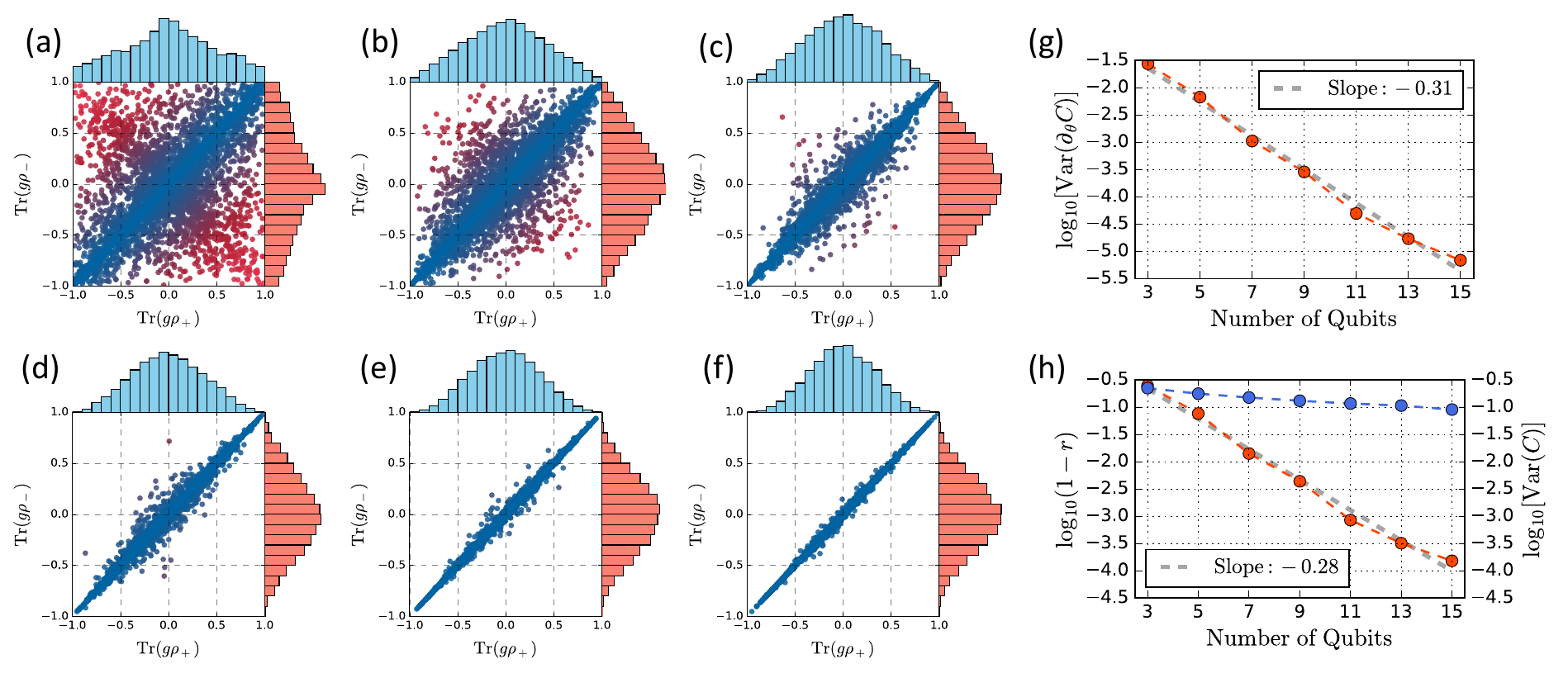}
    \caption{Hierarchical tree-circuit trainability diagnostics via parameter-shift statistics.
    (a)-(f) present scatter plots of shifted terms in circuits containing $3,5,7,9,11$, and $13$ qubits, respectively.
    (g) illustrates how gradient variance vanishes with increasing number of qubits, with a gray dashed line representing the linear fit in logarithmic scale.
    (h) displays the comparison between $(1-r)$ and the observable variance, accompanied by a fitted straight line for $(1-r)$ in logarithmic scale as well.}
    \label{fig:res-tree-circ}
\end{figure*}

\begin{example}\label{ex:1}
   For the circuit configuration shown in Fig.~\ref{fig:mid-circ-example1}, let $U$ and $V$ be unitary 2-design and assume the dimensional constraint $d_2 \leq d_3/2$. Under these conditions, the Pearson correlation coefficient admits the bound $1-r = \mathcal{O}(d_2^2/d_1)$.
\end{example}

Detailed calculation can be found in Methods. As shown in Example~\ref{ex:1}, the Pearson correlation coefficient satisfies $1 - r = \mathcal{O}(d_2^2 / d_1)$, which leads to barren plateaus with $d_1$ increasing exponentially. Meanwhile, the variance associated with the observable scales as $\mathcal{O}\left(d_2^{-2} d_3^{-1}\right)$ and is independent of $d_1$. This indicates that the component of the gradient variance induced by $d_1$ cannot be inferred from the measurement statistics of the observable alone.
{This example gives a tunable separation between the two mechanisms.}
{For fixed $d_2$, increasing $d_1$ leads to $r\rightarrow 1$ without changing the observable variance, whereas increasing $d_3$ suppresses the observable variance.}
{The value $d_1$ therefore isolates information-loss-induced suppression, while $d_3$ mainly contributes to observable concentration.}

In machine learning applications, training data is commonly divided into mini-batches to enhance computational efficiency and model robustness. This mini-batch learning approach has proven effective in both classical and quantum domains \cite{chen2021FederatedQuantumMachineLearning,sajjan2022QuantumMachineLearningChemistryPhysics}. We extend \cref{obs:mid-circ} by introducing a decomposition of a Hermitian operator into Pauli batches. Let $\{\mathcal{P}_j\}_{j=1}^K$ represent a partition of the Pauli set $\mathcal{P}$ into $K$ disjoint subsets. For a Hermitian operator $O$, we define its $K$-Pauli batch decomposition as $O = \sum_{j=1}^K B_j$, where the $j$-th batch $B_j = \sum_{g \in \mathcal{P}_j} c_g g$ comprises Pauli operators $g$ with corresponding coefficients $c_g$. This decomposition yields the following corollary:

\begin{corollary}[Batch-decomposed form of \Cref{obs:mid-circ}]
    {For the cost function in the form $C = \Tr\left[UR_G(\theta)V\rho_0V^\dg R_G(\theta)^\dg U^\dg O\right]$ with independently sampled $U,V\in\cE_{\rm LS}$ and the Hamiltonian in the $K$-Pauli decomposition form, i.e., $H = \sum_{j=1}^K B_j$, the variance of gradient can be bounded as}
    \begin{align*}
        \Var(\partial_\theta C) \le \sum_{j=1}^K u^2 \ltwo{B_j}^2 2^{|\lcs_j|+2} \bE_{V \sim \cE_{\rm LS}}~D_{\rm HS}^2[\varrho_{\lcs_j}(V)],
    \end{align*}
    where $\lcs_j$ is the light-cone subspace of $j$-th batch, i.e., $\lcs_j=\supp(U^\dg B_j U)$.
    \label{cor:batch_decomp_form}
\end{corollary}
A detailed proof of this corollary is provided in Methods.

\section*{BPs in Correlated Layers}

{We next give a complementary correlated-layer example.}
{Specifically, we consider the echo-type circuit illustrated in \cref{fig:mid-circ-example2}, where the backward block is fixed to be $U^\dagger$ and is therefore maximally correlated with the forward block $U$.}
{This setting is reminiscent of information scrambling \cite{xuScramblingDynamicsOutofTimeOrderedCorrelatorsQuantumManyBody2024}: the parameter perturbation remains globally present, but it is spread over the operator space so that a fixed $U$-independent observable can have a small overlap with it.}
{Thus, unlike the independent-layer information-loss setting above, the suppression here does not arise from discarding degrees of freedom outside the measurement light cone, but from scrambling-induced indistinguishability of the two parameter-shift evaluations.}
\begin{figure}[h]
    \centering
    \includegraphics[width=0.7\linewidth]{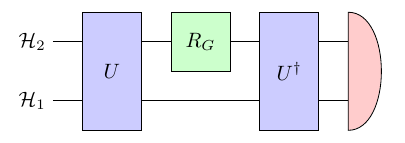}
    \caption{Circuit models to demonstrate the impact of information scrambling on BPs. The measurement is applied within the joint space $\cH_1\ot\cH_2$.}
    \label{fig:mid-circ-example2}
\end{figure}
{We consider a circuit model depicted in \Cref{fig:mid-circ-example2}. The model comprises a parameterized gate $R_G(\theta)$ sandwiched between two correlated unitary blocks, $U$ and $U^{\dagger}$, that operate on a $d$-dimensional quantum system. Such echo-type architectures are widely applied in tasks such as unknown unitary learning, scrambling-based quantum metrology \cite{kobrin2024UniversalProtocolQuantumEnhancedSensingInformationScramblinga}, and quantum autoencoders \cite{romero2017QuantumAutoencodersEfficientCompressionQuantumData,bondarenko2020QuantumAutoencodersDenoiseQuantumData}.}
{For clarity we fix the probed parameter at $\theta=0$ (so $R_G(\theta)=\id$) while averaging over the random unitary $U$ (or, equivalently, over the remaining circuit parameters).}
{This isolates scrambling-driven indistinguishability in the parameter-shift evaluations and should not be interpreted as focusing on an isolated critical point of a fixed deterministic landscape.}
{For this circuit, the observable second moment remains non-concentrated; the calculation in Methods then shows that the Pearson correlation factor satisfies $1-r=\mathcal{O}(1/d)$.}

{Equivalently, since $d=2^n$, the parameter-shift signal becomes exponentially less detectable as the system size grows. Our findings are summarized as follows.}

\begin{example}\label{ex:2}
   {For the circuit architecture depicted in Fig.~\ref{fig:mid-circ-example2} with $U$ drawn from a unitary $4$-design, the Pearson correlation coefficient satisfies $1 -r = \mathcal{O}\left( \frac{1}{d}\right)$, and the corresponding gradient variance scales as $1 - r$.}
\end{example}

The detailed calculation is presented in Methods.
Thus the gradient variance vanishes at the same order as $1-r$, while the observable second moment remains non-concentrated.
The example shows within \cref{eq:norm-grad} that a perturbation which remains globally present can be scrambled over a high-dimensional operator space, so that a fixed $U$-independent Pauli observable, even one with full support, may fail to distinguish the two parameter-shift evaluations.
Accordingly, the example provide a complementary perspective to the independent-layer information-loss setting: the suppression arises from scrambling-induced indistinguishability rather than from discarding degrees of freedom outside the measurement light cone.

\section*{Numerical Simulation}
{We numerically test our theoretical predictions using the hierarchical tree circuit architecture in \Cref{fig:tree-circ}, which is inspired by quantum convolutional neural network (QCNN) constructions~\cite{pesahAbsenceBarrenPlateausQuantumConvolutionalNeural2021} but has linear depth in the number of qubits.}
{Our purpose is therefore not to revisit QCNN trainability guarantees, but to use this controlled linear-depth hierarchical architecture to isolate the mid-circuit correlation mechanism ($r\to 1$) predicted by \cref{eq:norm-grad}.}
{The example was chosen using the information-loss criterion above: the central parameter has access to a growing upstream environment as $n$ increases, while the final observable remains confined to the last register.}
{In this geometry, the perturbation generated by the central gate is expected to leak into degrees of freedom that are not probed by the final measurement, so information loss should dominate.}

As illustrated in \Cref{fig:tree-circ}, the circuit model comprises $L$ layers of single-qubit gates $R_P^{(i,j)}(x)=\exp\left(-i x P/2\right)$, where $P$ is randomly selected from Paulis $\{X,Y,Z\}$, and two-qubit gates $R_{XX}^{(i,j)}(y)=\exp(-i y X\ot X)$. 
The indices $(i,j)$ on the superscript indicate that the gate is applied at $i$-th qubit and $j$-th layers. The complete circuit evolution can be written as $U_P^{(L+1)}\prod_{k=1}^{L} U_{XX}^{(k)} U_P^{(k)}$, where $U_P^{(k)}=\prod_{j=1}^{L-k+2} R_P^{(j,k)}$ and $U_{XX}^{(k)}=\prod_{j=1}^{L-k+1}R_{XX}^{(j,k)}$.
We let the observable be the Pauli $Z$ acting on the last register and the input state be the zero state $\ket{00\dots 0}$.
To analyze the mid-circuit dynamics, we focus on the derivative of the central single-qubit gate, i.e., $R_P^{(L/2,L/2)}$, and fix the number of layers $L=n-1$ with $n$ being the number of qubits.

{We sample all gate parameters independently and uniformly from $[0,2\pi)$ and evaluate the gradient via the parameter-shift rule, which yields two shifted terms $\Tr[g\rho_{+}]$ and $\Tr[g\rho_{-}]$.}
The results are shown in \Cref{fig:res-tree-circ}.
Panels (a)--(f) display scatter plots of $\Tr[g\rho_{+}]$ versus $\Tr[g\rho_{-}]$ (with marginal histograms) for increasing $n$.
{As $n$ grows, the two shifted terms become increasingly linearly dependent, approaching $\Tr[g\rho_{+}]\approx \Tr[g\rho_{-}]$, which is quantified by the Pearson correlation coefficient $r\to 1$.}
Panels (g) and (h) report the gradient variance, $r$, and the observable variance on logarithmic scales.
{Notably, the decay of $\log_{10}(1-r)$ closely matches the decay of the gradient variance, whereas the observable variance does not exhibit a comparable decrease.}
{This behavior indicates barren plateaus without OC, consistent with gradient suppression driven by mid-circuit information loss through the increasing correlation between the two shifted evaluations.}

\section*{Conclusion and Discussion}

{We have shown that observable concentration does not fully explain trainability issues in PQCs: the gradient variance naturally decomposes into an end-of-circuit observable second-moment term and a parameter-shift distinguishability term $(1-r)$.}
{In particular, we identify two mid-circuit settings in which BPs can arise without observable concentration: independent forward--backward layers, where information loss provides the main mechanism analyzed in our bounds and numerical example; and correlated forward--backward layers, where an echo-type scrambling construction gives a complementary correlation-factor suppression result.}
{This separation reveals a subtle but practically important regime: the objective may display noticeable fluctuations across random parameters, yet gradients with respect to specific parameters can be exponentially suppressed because the corresponding perturbations become inaccessible to the final measurement or are scrambled into directions poorly resolved by the chosen observable.}
{From this viewpoint, BPs reflect a limitation of the circuit's ability to transmit parameter dependence to the measured degrees of freedom, rather than solely a property of the output statistics.}
{Because $r$ can be estimated from the same parameter-shift evaluations used to compute gradients, it provides an experimentally accessible diagnostic of sensitivity loss that is complementary to end-of-circuit variance (or purity) diagnostics.}
{This perspective is also complementary to Lie-algebraic, purity- and entanglement-based, and global-cost approaches: those frameworks characterize accessible operator directions, reduced-state concentration, or observable concentration caused by the support size of the cost, whereas $(1-r)$ captures when such directions become practically indistinguishable due to mid-circuit information dynamics.}

Our results also offer a complementary perspective on parameter redundancy and effective model capacity~\cite{mehendale2023ExploringParameterRedundancyUnitaryCoupledClusterAnsatze,haug2021CapacityQuantumGeometryParametrizedQuantumCircuits}. 
Parameters whose influence is lost or scrambled contribute little to optimization, effectively reducing the dimensionality of the trainable parameter space even when the nominal parameter count is large.

Several directions follow naturally. First, it will be valuable to develop practical diagnostics that detect mid-circuit information loss or excessive scrambling during training, and to translate these diagnostics into architecture or initialization principles that preserve parameter sensitivity. Second, understanding the interplay with noise remains an open issue: realistic noise may further erase propagating perturbations, but it may also suppress scrambling or modify the effective causal structure in ways that are not captured by idealized models. Clarifying these effects is important for predicting trainability on hardware and for designing scalable variational algorithms that can realize reliable performance improvements.

\section*{Methods}
\subsection*{Proof of \Cref{thm:obs_conc}}\label{app:proof-end-circ-statis}

\begin{proof}
    We first write the output state in the Pauli basis:
    \begin{align*}
        \rho = U(\bthe) \rho_0 U(\bthe)^\dg =  \frac{1}{d}\left(\id+\sum_{k=1}^{d^2-1} \lambda_k g_k\right),
    \end{align*}
    where $\id$ is the identity operator, $\{g_k\}$ are the Pauli operators, and $\lambda_k\in \mathbf{R}$ are expectation values of the Pauli operators, i.e., $\lambda_k = \Tr(\rho g_k)$.
    Then we have
    \begin{align*}
        \bE_{U\sim\cE}~\lambda_k^2=\bE_{U\sim\cE}~[\Tr(g_k U\rho_0 U^\dg)]^2,
    \end{align*}
    which can be further written as follows, according to the locally scrambling properties:
    \begin{align*}
        \bE_{U\sim\cE}~\Tr(g_k^{\ot 2}\rho^{\ot 2})=\bE_{U\sim\cE}~\Tr(C^{\dg\ot2} g_k^{\ot 2} C^{\ot 2} \rho^{\ot 2}),
    \end{align*}
    with $C$ being a tensor product of arbitrary single-qubit Clifford operations.
    The single-qubit Clifford operations can map the operator $g$ to arbitrary Pauli operators while preserving its support. We then have
    \begin{align*}
        \bE_{U\sim\cE}~\Tr(\rho^{2}) &= \frac{1}{d}\left[1+\sum_{k=1}^{d^2-1} \bE_\cE~({\lambda_k^2}) \right]\\
        &= \frac{1}{d}\left[1 + \sum_{S\subset [n]} \sum_{k:~\supp(g_k)=S}\bE_\cE~({\lambda_k^2})\right],
    \end{align*}
    where all the $\bE_\cE~({\lambda_k^2})$ with $\supp(g_k)=S$ must has the same average.
    We then have
    \begin{align*}
        \bE_{U\sim\cE}~\Tr(\rho^{2}) = \frac{1}{d}\left[1 + \sum_{S\subset [n]}3^{|S|}\bE_\cE~({\lambda_S^2})\right].
    \end{align*}
    Let $P_A$ be a unit Pauli operator supported on $A$, and we have
    \begin{align}
        \bE_{U\sim\cE}~\Tr(\rho_A^{2}) &= \frac{1}{d_A}\left[1 + \sum_{S\subset A}3^{|S|}\bE_\cE~({\lambda_S^2})\right]\nonumber\\
        &\ge \frac{1}{d_A} \left[1 + 3^{|A|}~\bE_\cE (\lambda_A^2)\right].
        \label{eq:methods-reduced-purity-bound}
    \end{align}
    Solving \Cref{eq:methods-reduced-purity-bound} for the unit-Pauli second moment gives
    \begin{align*}
        \bE_{U\sim \cE}~[\Tr(P_A\rho)]^2 \le \frac{d_A \bE_{U\sim\cE}~\Tr(\rho_A^{2}) - 1}{3^{|A|}}.
    \end{align*}
{Since $d_A=2^{|A|}$ and $D_{\rm HS}^2(\rho_A)=\Tr[\rho_A^2]-2^{-|A|}$, the right-hand side equals $\left(2/3\right)^{|A|}\bE_{U\sim\cE}D_{\rm HS}^2(\rho_A)$.}
{For a scaled Pauli term $g=\alpha P_A$, $\Tr(g\rho)=\alpha\Tr(P_A\rho)$ and $\ltwo{g}=|\alpha|$, which gives \Cref{eq:end-circ-statis}.}
\end{proof}
\begin{remark}
Using the orthogonality of the locally scrambling unitary ensemble, one can decompose a general Hermitian operator as \(H = \sum_j c_j g_j\). Then
\(\bE_{U\sim \cE_{\rm LS}} [\Tr(\rho H)]^2\) can be written as
\begin{align}
    \bE_{U\sim \cE_{\rm LS}} ~\sum_{j,k} c_j c_k
    \Tr\left(\rho_0^{\ot 2} U^{\dg \ot 2} (g_j \ot g_k) U^{\ot 2}\right)\nonumber .
\end{align}
By \cref{def:local-scramb}, the cross terms vanish due to the Pauli-mixing property
\cite{angrisaniClassicallyEstimatingObservablesNoiselessQuantumCircuits}. Consequently, the expression reduces to a linear combination of
\(\bE_{U\sim \cE_{\rm LS}}[\Tr(\rho g_j)]^2\).
Therefore, \cref{thm:obs_conc} is sufficient to reflect the concentration behavior of general Hermitian operators.
\end{remark}

\subsection*{Proof of \Cref{eq:var_information_loss}}
\label{app:proof_mid_circ}
\begin{proof}
    Let $\tilde{\rho}_V = G \rho_V G - \rho_V$ and the gradient variance can be written as
    \begin{align}
        \Var_{U,V \sim \cE}~\partial_\theta C &= c^2 u^2 \Tr\left[\bE_{U, V}~g^{\ot 2} (U\tilde{\rho}_VU^\dg)^{\ot 2}\right]\nonumber\\
        &=c^2 u^2 \Tr\left[\bE_{U, V}~(U^\dg gU)^{\ot 2} \tilde{\rho}_V^{\ot 2}\right].
        \label{eq:methods-gradient-variance-tilde}
    \end{align} 
    We let $\lcs$ denote the light-cone subspace at the snapshot $\rho_V$, which can be defined by the support of the operator $U^\dg g U$.
    We can then relax the bound by introducing an operator $M_{\lcs}$ with support $\lcs$:
    \begin{align}
        \Tr&\left[\bE_{U, V}~(U^\dg gU)^{\ot 2} \tilde{\rho}_V^{\ot 2}\right] \nonumber\\
        &\le \sup_{\ltwo{M_{\lcs}}\le 1}\ltwo{g}^2~\Tr\left(M_{\lcs}^{\ot 2}\bE_{V}~ \tilde{\rho}_V^{\ot 2}\right)\nonumber\\
        &=\sup_{\ltwo{M_{\lcs}}\le 1}\ltwo{g}^2~\Tr\left(M_{\lcs}^{\ot 2}\bE_{V}~ \Tr_{\neq \lcs} \tilde{\rho}_V^{\ot 2}\right)\nonumber\\
        &\le\sup_{\ltwo{M_{\lcs}'}\le 1}\ltwo{g}^2~\Tr\left(M_{\lcs}'\bE_{V}~ \Tr_{\neq \lcs} \tilde{\rho}_V^{\ot 2}\right)\nonumber\\
        &= \ltwo{g}^2 \lone{\bE_{V}~ \Tr_{\neq \lcs} \tilde{\rho}_V^{\ot 2}}.
        \label{eq:methods-light-cone-relaxation}
    \end{align}
    The second inequality arises from relaxing $M_{\lcs}^{\ot 2}$ to $M_{\lcs}'$, where $M_{\lcs}'$ is an operator with support on $\lcs^{\ot 2}$.
    Combining \Cref{eq:methods-gradient-variance-tilde,eq:methods-light-cone-relaxation} gives \Cref{eq:var_information_loss}.
\end{proof}

\subsection*{Proof of \Cref{obs:mid-circ}}

\begin{proof}
According to \Cref{eq:var_information_loss}, the gradient variance can be upper bounded by evaluating the trace norm of the quantity $ \bE_V~\Tr_{\neq \lcs}~\tilde{\rho}_V^{\ot 2}$.

We first expand $\tilde{\rho}_V^{\ot 2}$:
\begin{align}
    \tilde\rho_{V}^{\ot 2} &= (G\rho_V G - \rho_V)^{\ot 2}\nonumber\\
    &= \frac{1}{d^2} \left[\sum_j \lambda_j \left(G g_j G - g_j\right)\right]^{\ot 2}\nonumber\\
    &= \frac{1}{d^2} \left(\sum_{j: [g_j,G]\neq 0} \lambda_j \times 2g_j\right)^{\ot 2}\nonumber\\
    &= \frac{4}{d^2} \sum_{i,j:[g_{i,j},G]\neq0}\lambda_i \lambda_j (g_i \ot g_j),
    \label{eq:methods-tilde-rho-expansion}
\end{align}
where $g_{i,j}$ is used as shorthand to indicate either $g_i$ or $g_j$, with a slight abuse of notation. 
Here we use the expanded form of $\rho_V=(\id + \sum_j g_j \lambda_j)/d$.
As the partial trace operation will eliminate most of the terms in $\{g_i\ot g_j\}$, \Cref{eq:methods-tilde-rho-expansion} can be simplified by defining the Pauli set $\cP_{\rm eff}:=\{g|g\in \cP_{\lcs}, gG+Gg=0\}$. 
We then have
\begin{align}
    \Tr_{\neq \lcs}~\tilde{\rho}_V^{\ot 2} = \frac{4}{4^{|\lcs|}} \sum_{i,j: g_{i,j}\in \cP_{\rm eff}}\lambda_i \lambda_j (g_i\ot g_j).
\label{eq:methods-lightcone-partial-trace}
\end{align}
Now we compute the ensemble average in \Cref{eq:methods-lightcone-partial-trace}.
{As the ensemble of $V$ is locally scrambling, we can ignore the cross terms as they do not contribute to the average:}
\begin{align}
    \bE_V~\Tr_{\neq \lcs}~\tilde{\rho}_V^{\ot 2} = \frac{4}{4^{|\lcs|}} \sum_{j: g_{j}\in \cP_{\rm eff}}g_j^{\ot 2}\bE_V~\lambda_j^2.
\label{eq:methods-effective-pauli-average}
\end{align}
We need to evaluate the trace norm that should be upper bounded by the Frobenius norm through the C-S inequality:
\begin{align}
    \lone{\bE_V~\Tr_{\neq \lcs}~\tilde{\rho}_V^{\ot 2}} \le 2^{|\lcs|} \fbnorm{\bE_V~\Tr_{\neq \lcs}~\tilde{\rho}_V^{\ot 2}}.
\label{eq:methods-trace-frobenius}
\end{align}
Using \Cref{eq:methods-effective-pauli-average}, we further expand it to obtain:
\begin{align}
    \fbnorm{\bE_V~\Tr_{\neq \lcs}~\tilde{\rho}_V^{\ot 2}} &= \frac{4}{4^{|\lcs|}} \fbnorm{\sum_{j: g_{j}\in \cP_{\rm eff}}g_j^{\ot 2}\bE_V~\lambda_j^2}\nonumber\\
    &= \frac{4}{4^{|\lcs|}} \sqrt{\sum_{j: g_{j}\in \cP_{\rm eff}} 4^{|\lcs|} \left(\bE_V~\lambda_j^2\right)^2}\nonumber\\
    &\le \frac{4}{4^{|\lcs|}}  \sum_{j:g_j \in \cP_{\rm eff}} 2^{|\lcs|} \bE_V~\lambda_j^2\nonumber\\
    &= 4 D_{\rm HS}^2 [\varrho_{\lcs}(V)].
\label{eq:methods-effective-pauli-frobenius}
\end{align}
Here $\varrho_\Gamma (V) = \id_\lcs / 2^{|\lcs|}+ \sum_{j:g_j \in \cP_{\rm eff}} \lambda_j~ g_j/2^{|\lcs|}$.
Substituting \Cref{eq:methods-trace-frobenius,eq:methods-effective-pauli-frobenius} into \Cref{eq:var_information_loss}, we conclude that
\begin{align*}
  \text{Var}_{U,V\sim \cE} \partial_{\theta_i}C &\leq c^2 u^2 \ltwo{g}^2 2^{|\Gamma|} \fbnorm{\bE_V~\Tr_{\neq \lcs} \tilde{\rho}_V^{\ot 2}}\\
  &\le c^2 u^2 \ltwo{g}^2 2^{|\lcs|+2}\bE_V~D_{\rm HS}^2[\varrho_{\lcs}(V)].
\end{align*}
\end{proof}

\subsection*{Calculation of \Cref{ex:1}}\label{app:ex1}
Consider unitary of a circuit acting on the composite Hilbert space $\cH_{1}\ot \cH_{2}\ot \cH_{3}$ with dimensions $d_1,d_2$, and $d_3$, having the form $U_{(2,3)} e^{-i G_{[2]} \theta} V_{(1,2)}$, where the indices $(i,j,\cdots)$ denote the labels of supports of the operators, i.e., $O_{(i,j,...)} = \id_{\neq i,j,...}\ot O$.
The measurement observable is $g_{[2,3]}$ and the generator of the parameterized gate $G_{[2]}$, where the indices $[i,j,\dots]$ denotes the support of the operator is a subset of the spaces labeled by $(i,j,\dots)$. Specifically, $\supp(X_{(i,j,\dots)}) = \cH_{i}\ot \cH_{j}\ot \dots$ and $\supp(X_{[i,j,\dots]}) \subset \cH_{i}\ot\cH_{j}\ot\dots$. 
In this example, $U$ and $V$ are chosen from 2-design unitary ensembles, which we denote as $U\sim\cE_\text{2-d}$ and $V\sim \cE_\text{2-d}$ respectively. 
For conciseness, we will ignore the notion of $\cE_\text{2-d}$ in the following calculations.

The quantum circuit in Example~\ref{ex:1} can be expressed as $U_{(2,3)} e^{-i G_{[2]} \theta} V_{(1,2)}$, where the subscript notation $A_{(i_1,\dots,i_\ell)}$ with $1 \leq i_1 < \cdots < i_\ell \leq 3$ and $1 \leq \ell \leq 3$ denotes an $n$-qubit operator that acts non-trivially only on Hilbert space $\{\cH_{i_1}, \dots, \cH_{i_\ell}\}$. Here, \emph{nontrivially} means that the operator acts as the identity on all other qubits. Similarly, $A_{[i_1,\dots,i_\ell]}$ denotes an operator supported on the Hilbert space $\{i_1, \dots, i_\ell\}$.

We use the Heisenberg picture and first calculate the average in $g$. According to Weingarten's calculus, we have
\begin{align*}
    \bE_{U}~U_{(2,3)}^{\dg\ot 2} g^{\ot 2}_{[2,3]} U_{(2,3)}^{\ot 2} = b_U \id + c_U \fl_{(2,3)},
\end{align*}
with the coefficients being
\begin{align}
    b_U = \frac{- 1}{d_{23}^2-1},~c_U =\frac{d_{23}}{d_{23}^2 - 1}\label{eq:app:ex1-1}.
\end{align}
To calculate the average in $[\Tr(\rho g)]^2$, we have
\begin{align*}
    \bE_{U, V}~ &V^{\dg\ot 2}_{(1,2)} U_{(2,3)}^{\dg\ot 2} g^{\ot 2}_{[2,3]} U_{(2,3)}^{\ot 2} V^{\dg \ot 2}_{(1,2)} \\
    &= b_U \id + c_U \bE_{V}~ V^{\dg \ot 2}_{(1,2)} \fl_{(2,3)} V^{\ot 2}_{(1,2)},
\end{align*}
where the second term reduces to
\begin{align*}
    \bE_{V}~ V^{\dg \ot 2}_{(1,2)} \fl_{(2,3)} V^{\ot 2}_{(1,2)} = b_V \fl_{(3)} + c_V \fl_{(1,2,3)},
\end{align*}
with the coefficients being
\begin{align}
    b_V = \frac{d_1^2 d_2 - d_2}{d_{12}^2 - 1},~ c_V = \frac{d_1d_2^2 - d_1}{d_{12}^2-1}\label{eq:app:ex1-2}.
\end{align}
Then we have
\begin{align}
    &\bE_{U,V}~[\Tr(\rho g)]^2 = \bE_{U,V}~\Tr(\rho^{\ot 2} g^{\ot 2}) \nonumber\\
    &= \Tr\left[\rho_0^{\ot 2}\left(b_U \id + c_U b_V \fl_{(3)}+c_U c_V \fl_{(1,2,3)}\right)\right] \nonumber\\
    & = b_U + c_U b_V \Tr[\rho_0^{\ot 2} \fl_{(3)}] + c_U c_V \Tr[\rho_0^{\ot 2}\fl_{(1,2,3)}] \nonumber\\
    & = b_U + c_U b_V \Tr\left[\rho^2_{0(3)}\right]+c_U c_V \Tr\left[\rho_0^2\right]\label{eq:app:ex1-3}.
\end{align}
Let $\dD = \sqrt{G_{[2]}} \ot \sqrt{G_{[2]}}^{\dg}$, we have
\begin{align*}
    &\bE_{U, V}~ V^{\dg\ot 2}_{(1,2)} \dD U_{(2,3)}^{\dg\ot 2} g^{\ot 2}_{[2,3]} \dD^\dg U_{(2,3)}^{\ot 2} V^{\dg \ot 2}_{(1,2)} \\
    &\quad= b_U \id + c_U \bE_{V}~ V^{\dg \ot 2}_{(1,2)} \dD \fl_{(2,3)} \dD^\dg V^{\ot 2}_{(1,2)},
\end{align*}
where the second term reduces to
\begin{align*}
    \bE_{V}~ V^{\dg \ot 2}_{(1,2)} \dD \fl_{(2,3)} \dD^\dg V^{\ot 2}_{(1,2)} = b_V' \fl_{(3)} + c_V' \fl_{(1,2,3)},
\end{align*}
with the coefficients being
\begin{align}
    b_V' = \frac{d_1^2 d_2}{d_{12}^2 - 1},~ c_V' = \frac{- d_1}{d_{12}^2 -1}\label{eq:app:ex1-4}.
\end{align}
We have
\begin{align}
    &\bE_{U,V}~ \Tr(\rho_+ g) \Tr(\rho_- g) = \bE_{U,V}~\Tr(\rho^{\ot 2} g^{\ot 2}) \nonumber\\
    &= \Tr\left[\rho_0^{\ot 2}\left(b_U \id + c_U b_V' \fl_{(3)}+c_U c_V' \fl_{(1,2,3)}\right)\right] \nonumber\\
    & = b_U + c_U b_V' \Tr[\rho_0^{\ot 2} \fl_{(3)}] + c_U c_V' \Tr[\rho_0^{\ot 2}\fl_{(1,2,3)}] \nonumber\\
    & = b_U + c_U b_V' \Tr\left[\rho^2_{0(3)}\right]+c_U c_V' \Tr\left[\rho_0^2\right]\label{eq:app:ex1-5}.
\end{align}

Combine \Cref{eq:app:ex1-1,eq:app:ex1-2,eq:app:ex1-3,eq:app:ex1-4,eq:app:ex1-5}, we have
\begin{align*}
    1-r 
    &\approx\frac{d_2^2 d_3 d_1 d_2}{d_1 d_2^3 d_3 - d_1^2 d_2^2 + d_1^2 d_2 d_3},
\end{align*}
where assume the initial state $\rho_0$ to be a product pure state. If, in addition, the dimensional constraint $d_2 \leq d_3/2$ holds, then the Pearson correlation coefficient satisfies
\begin{align*}
    1-r = \mathcal{O}\left(\frac{d_2^2}{d_1}\right).
\end{align*}

\subsection*{Proof of Corollary \ref{cor:batch_decomp_form}}
\label{app:proof_corollary}

\begin{proof}[Proof of Corollary \ref{cor:batch_decomp_form}]
    Let $W=UR_G(\theta)V$ and rewrite the output state as $\rho_W$. 
    We first write the cost function in the Pauli decomposition form $C = \Tr(\rho_W H) = \sum_{j=1}^{L} c_j C_j$, where $C
    _j=\Tr(\rho_W g_j)$ and $\{g_j\}$ are Pauli operators.
    The variance of the gradient can also be expanded component-wise accordingly:
    \begin{align}
        &\Var_{W\sim\cE}\left[\sum_j c_j \partial_{\theta} \Tr(\rho_W g_j)\right] \nonumber\\
        &= \sum_{j,k} u^2 \bE_{W \sim \cE_{\rm LS}}~ c_j c_k \Tr[(\rho_+ - \rho_-) g_j]\Tr[(\rho_+ - \rho_-) g_k]\nonumber\\
        &= \sum_{j,k} u^2 \bE_{W \sim \cE_{\rm LS}}~ c_j c_k \Tr\left[(\rho_+ - \rho_-)^{\ot 2} g_j\ot g_k\right].
        \label{eq:methods-batch-variance-expansion}
    \end{align}
    Here $\rho_\pm$ are parameter shifts of $\rho_W$ on the parameter $\theta$.
    By switching to the Heisenberg picture and applying the trick of Pauli-mixing \cite{angrisaniClassicallyEstimatingObservablesNoiselessQuantumCircuits}, we have 
    \begin{align*}
        &\bE_{W\sim \text{Cl}(2)^{\ot n}}~ W^{\dg\ot 2} (g_j\ot g_k) W^{\ot 2}
        = 0 \quad\text{if } j\neq k,
    \end{align*}
    with $W$ being a tensor product of arbitrary single-qubit Clifford gates.
    This leads to
    \begin{align}
        \bE_{W\sim \text{Cl}(2)^{\ot n}}~ W^{\dg\ot 2} (B_j\ot B_k) W^{\ot 2}
        = 0 \quad\text{if } j\neq k.
        \label{eq:methods-batch-cross-terms}
    \end{align}
    We then have
    \begin{align}
        \Var_{U\sim\cE}~ \partial_{\theta_i}C &= \sum_{j=1}^K u^2 \Tr\left[(\rho_+ - \rho_-)^{\ot 2} B_j^{\ot 2}\right].
        \label{eq:methods-batch-gradient-decomposition}
    \end{align}
    Applying \Cref{eq:methods-batch-cross-terms} to \Cref{eq:methods-batch-variance-expansion} gives the Pauli-batch decomposition in \Cref{eq:methods-batch-gradient-decomposition}.
    This reduces the variance to a linear combination of the Pauli batches, where \Cref{obs:mid-circ} can be applied individually.
    This completes the proof.
\end{proof}

\subsection*{Calculation of \Cref{ex:2}}\label{app:ex2}

{We compute the correlation factor $r$ for the echo-type correlated-layer scrambling circuit at $\theta=0$.}
The purpose of the calculation is to show that the numerator and denominator of $r$ have the same leading fourth-moment contribution, while their first possible difference is at most order $d^{-1}$.

Let $1,2,3,4$ denote the four registers of the $d$-dimensional Hilbert space.
In the untransposed vectorized expression, the two shift patterns are represented by
\begin{align*}
    P' &= \sqrt{G} \ot \sqrt{G}^* \ot \sqrt{G} \ot \sqrt{G}^*,\\
    Q' &= \sqrt{G} \ot \sqrt{G}^* \ot \sqrt{G}^\dg \ot \sqrt{G}^T.
\end{align*}
After applying the partial transpose on the registers labeled by $2$ and $4$, we obtain
\begin{align*}
    P &= {P'}^{T_{24}} = \sqrt{G} \ot \sqrt{G}^\dg \ot \sqrt{G} \ot \sqrt{G}^\dg,\\
    Q &= {Q'}^{T_{24}} = \sqrt{G} \ot \sqrt{G}^\dg \ot \sqrt{G}^\dg \ot \sqrt{G}.
\end{align*}
With this notation, the Pearson coefficient can be written as
\begin{align}
    r &= \frac{\mathcal{M}_Q}{\mathcal{M}_P},\nonumber\\
    \mathcal{M}_X &:= \bE_{U\sim\cE}~\bra{\psi}^{\ot 4} [U^{\ot 4} X U^{\dg \ot 4}]^{T_{24}}\Ket{g}^{\ot 2},\nonumber\\
    &\hspace{1.5cm} X\in\{P,Q\}.
    \label{eq:methods-ex2-r-ratio}
\end{align}
Because \Cref{eq:methods-ex2-r-ratio} depends only on the fourth moment of $U$, the Haar calculation below applies equally to any exact unitary $4$-design.

For $X\in\{P,Q\}$, Weingarten calculus gives
\begin{align}
    \bE_{U\sim \mu_H} (U^{\ot 4} X U^{\dg \ot 4}) = \sum_{\pi, \sigma \in S_4}  \wg(\pi^{-1} \sigma) \Tr\left(V_\sigma X\right)V_\pi,
    \label{eq:methods-ex2-weingarten}
\end{align}
where $V_\pi$ is the permutation representation of $\pi\in S_4$ and $\wg$ is the Weingarten function.
The leading part of \Cref{eq:methods-ex2-weingarten} comes from $\pi^{-1}\sigma=\mathrm{id}$, so it can be written as
\begin{align}
    \bE_{U\sim \mu_H} (U^{\ot 4} X U^{\dg \ot 4}) = \frac{1}{d^4} \sum_{\pi\in S_4}\Tr(V_\pi X) V_{\pi^{-1}} + R_X,
    \label{eq:methods-ex2-leading-moment}
\end{align}
where $R_X$ collects terms whose Weingarten coefficients are smaller by at least one power of $d$.
In \Cref{eq:methods-ex2-leading-moment}, the Pauli structure implies $|\Tr(V_\pi X)|\le d^{\#\mathrm{cycle}(\pi)}$.
The identity permutation is therefore dominant, and for both $P$ and $Q$ it gives $\Tr(V_{\rm id}X)=d^4/4$.
Thus
\begin{align}
    \bE_{U\sim\mu_H}~(U^{\ot 4} P U^{\dg \ot 4}) &= \frac{1}{4}\id + O(d^{-1}),\nonumber\\
    \bE_{U\sim\mu_H}~(U^{\ot 4} Q U^{\dg \ot 4}) &= \frac{1}{4}\id + O(d^{-1}).
    \label{eq:methods-ex2-identity-sector}
\end{align}

The estimate in \Cref{eq:methods-ex2-identity-sector} is not by itself sufficient, because the final contraction in \Cref{eq:methods-ex2-r-ratio} is taken after the partial transpose $T_{24}$.
This contraction can multiply selected permutation terms by a factor of $d$, as is visible from $\Braket{g|g}=d$.
The only permutation contractions with this $O(d)$ enhancement are
\begin{align}
    \bra{\psi}^{\ot 4}V_{\pi_1}^{T_{24}}\Ket{g}^{\ot 2}
    &= \includegraphics[width=0.23\linewidth, valign=c]{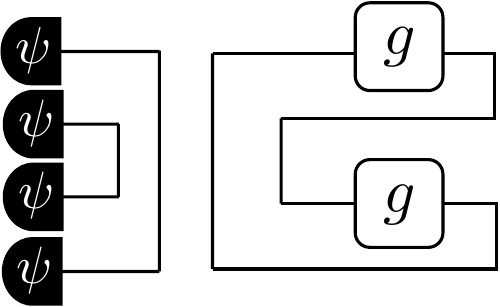} = d,\quad \pi_1=(14)(23),\nonumber\\
    \bra{\psi}^{\ot 4}V_{\pi_2}^{T_{24}}\Ket{g}^{\ot 2}
    &= \includegraphics[width=0.23\linewidth, valign=c]{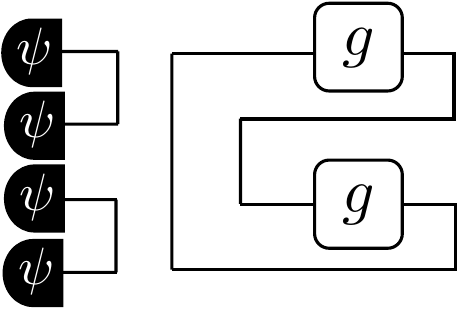} = d,\quad \pi_2=(1234).
    \label{eq:methods-ex2-large-contractions}
\end{align}
All other partial-transposed permutation contractions are lower order or vanish because a single traceless Pauli operator $g$ appears in a closed loop.
The diagrams corresponding to $\pi_1$ and $\pi_2$ are obtained by partially transposing the middle replicas:
\begin{align*}
    V_{\pi_1} &= \left(~~\includegraphics[width=0.1\linewidth, valign=c]{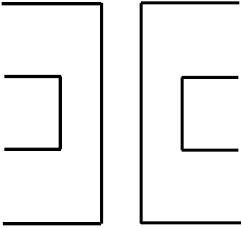}~~\right)^{T_{24}} = \includegraphics[width=0.07\linewidth, valign=c]{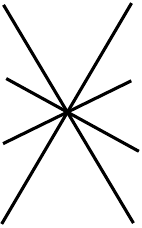} \Rightarrow (14)(23),\\
    V_{\pi_2} &= \left(~~\includegraphics[width=0.09\linewidth, valign=c]{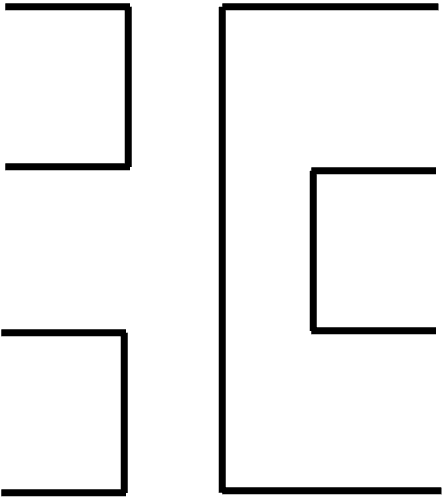}~~\right)^{T_{24}} = \includegraphics[width=0.07\linewidth, valign=c]{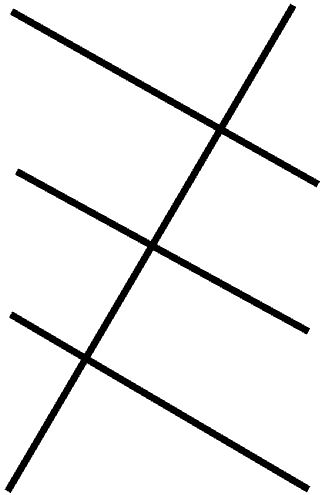} \Rightarrow (1234).
\end{align*}
The remaining tensor contraction fails to define a valid permutation after $T_{24}$,
\begin{equation*}
    \includegraphics[width=0.15\linewidth, valign=c]{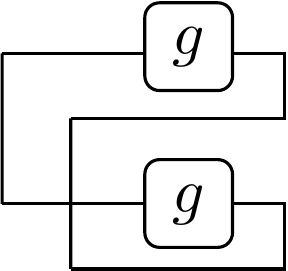},
\end{equation*}
and hence it does not create another $O(d)$ contribution.

We now estimate the two surviving corrections in \Cref{eq:methods-ex2-large-contractions}.
The large-$d$ expansion of the Weingarten function is \cite{collins2022WeingartenCalculusa}
\begin{align*}
    \wg(\pi^{-1}\sigma) = \frac{(-1)^{|\pi^{-1} \sigma|}}{d^{k+|\pi^{-1} \sigma|}} \sum_{j=0}^\infty \frac{\vec{W}_j(\pi,\sigma)}{d^{2j}},
\end{align*}
where $|\pi|=k-\#\mathrm{cycle}(\pi)$ and $\vec{W}_j(\pi,\sigma)$ counts weakly monotone walks on $S_k$ from $\pi$ to $\sigma$ of length $|\pi^{-1}\sigma|+2j$.
The leading coefficient is determined by the cycle type $\alpha$ of $\pi^{-1}\sigma$:
\begin{align*}
    \vec W_0 (\pi, \sigma) = \prod_i \frac{1}{\alpha_i} \binom{2\alpha_i-2}{\alpha_i}.
\end{align*}
In the present estimate, $\sigma=\mathrm{id}$ gives the dominant trace factor $\Tr(V_{\rm id}P)=d^4/4$.
For $\pi_1=(14)(23)$, the cycle type is $(2,2)$ and $\vec W_0(\pi_1,\mathrm{id})=1/4$.
For $\pi_2=(1234)$, the cycle type is $(4)$ and $\vec W_0(\pi_2,\mathrm{id})=15/4$.
Define the contribution associated with a permutation $\pi$ by
\begin{align*}
    C_{\pi}^{(P)}
    := \wg(\pi^{-1}\mathrm{id})\Tr(V_{\rm id}P)
    \bra{\psi}^{\ot 4} V_{\pi}^{T_{24}} \Ket{g}^{\ot 2}.
\end{align*}
Using \Cref{eq:methods-ex2-large-contractions}, the two relevant contributions are
\begin{align}
    C_{\pi_1}^{(P)}
    &= \frac{1}{4d^6}\cdot\frac{d^4}{4}\cdot d + O(d^{-3})\nonumber\\
    &= \frac{1}{16d}+O(d^{-3}),\nonumber\\
    C_{\pi_2}^{(P)}
    &= -\frac{15}{4d^7}\cdot\frac{d^4}{4}\cdot d + O(d^{-4})\nonumber\\
    &= -\frac{15}{16d^2}+O(d^{-4}).
    \label{eq:methods-ex2-surviving-corrections}
\end{align}
The estimates in \Cref{eq:methods-ex2-surviving-corrections} have the same scaling for $Q$.
Therefore both $\mathcal{M}_P$ and $\mathcal{M}_Q$ in \Cref{eq:methods-ex2-r-ratio} have the same order-one identity contribution and differ only through $O(d^{-1})$ corrections.
Consequently,
\begin{align}
    r = \frac{A_0+O(d^{-1})}{A_0+O(d^{-1})}=1+O(d^{-1}),
    \label{eq:methods-ex2-r-scaling}
\end{align}
where $A_0$ is the common order-one contribution from the identity sector.
Equivalently, $1-r=O(d^{-1})$.
{Since the observable second moment is non-concentrated for this circuit, \Cref{eq:norm-grad,eq:methods-ex2-r-scaling} implies that the gradient variance scales as $1-r$, which proves \Cref{ex:2}.}



\section*{DATA AVAILABILITY}
Data generated and analyzed during current study are available from the corresponding author upon reasonable request.

\section*{CODE AVAILABILITY}
Code used to generate data in this study are available from the corresponding author upon reasonable request.

\section*{Acknowledgements}
This work is supported by the National Natural Science Foundation of China Grant No.~12405014, Guangdong Provincial Quantum Science Strategic Initiative (project GDZX2403008), and the General Research Fund (GRF) grant 17303923.

\section*{AUTHOR CONTRIBUTIONS}
X.-W.L. and M.-H.Y. conceived the project. Z.‑S.L. and B.W. drafted the manuscript text and derived the theoretical results. X.‑W.L. contributed to the analysis and discussion of the results. M.‑H.Y. revised the manuscript. All authors approved the final version of the manuscript.

\section*{COMPETING INTERESTS}
The authors declare no competing interests.

\bibliography{reference}

@article{chen2021FederatedQuantumMachineLearning,
  title = {Federated {{Quantum Machine Learning}}},
  author = {Chen, Samuel Yen-Chi and Yoo, Shinjae},
  year = {2021},
  month = apr,
  journal = {Entropy},
  volume = {23},
  number = {4},
  pages = {460},
  publisher = {Multidisciplinary Digital Publishing Institute},
  issn = {1099-4300},
  doi = {10.3390/e23040460},
  copyright = {http://creativecommons.org/licenses/by/3.0/},
  langid = {english}
}

@article{sajjan2022QuantumMachineLearningChemistryPhysics,
  title = {Quantum Machine Learning for Chemistry and Physics},
  author = {Sajjan, Manas and Li, Junxu and Selvarajan, Raja and Hari~Sureshbabu, Shree and Suresh~Kale, Sumit and Gupta, Rishabh and Singh, Vinit and Kais, Sabre},
  year = {2022},
  journal = {Chemical Society Reviews},
  volume = {51},
  number = {15},
  pages = {6475--6573},
  publisher = {Royal Society of Chemistry},
  doi = {10.1039/D2CS00203E},
  langid = {english}
}

@article{caro2023OutofdistributionGeneralizationLearningQuantumDynamicsa,
  title = {Out-of-Distribution Generalization for Learning Quantum Dynamics},
  author = {Caro, Matthias C. and Huang, Hsin-Yuan and Ezzell, Nicholas and Gibbs, Joe and Sornborger, Andrew T. and Cincio, Lukasz and Coles, Patrick J. and Holmes, Zo{\"e}},
  year = {2023},
  month = jul,
  journal = {Nature Communications},
  volume = {14},
  number = {1},
  pages = {3751},
  publisher = {Nature Publishing Group},
  issn = {2041-1723},
  doi = {10.1038/s41467-023-39381-w},
  copyright = {2023 The Author(s)},
  langid = {english}
}

@misc{kobrin2024UniversalProtocolQuantumEnhancedSensingInformationScramblinga,
  title = {A {{Universal Protocol}} for {{Quantum-Enhanced Sensing}} via {{Information Scrambling}}},
  author = {Kobrin, Bryce and Schuster, Thomas and Block, Maxwell and Wu, Weijie and Mitchell, Bradley and Davis, Emily and Yao, Norman Y.},
  year = {2024},
  month = nov,
  number = {arXiv:2411.12794},
  publisher = {arXiv},
  doi = {10.48550/arXiv.2411.12794}
}

@article{bondarenko2020QuantumAutoencodersDenoiseQuantumData,
  title = {Quantum {{Autoencoders}} to {{Denoise Quantum Data}}},
  author = {Bondarenko, Dmytro and Feldmann, Polina},
  year = {2020},
  month = mar,
  journal = {Physical Review Letters},
  volume = {124},
  number = {13},
  pages = {130502},
  publisher = {American Physical Society},
  doi = {10.1103/PhysRevLett.124.130502}
}

@article{romero2017QuantumAutoencodersEfficientCompressionQuantumData,
  title = {Quantum Autoencoders for Efficient Compression of Quantum Data},
  author = {Romero, Jonathan and Olson, Jonathan P and {Aspuru-Guzik}, Alan},
  year = {2017},
  month = aug,
  journal = {Quantum Science and Technology},
  volume = {2},
  number = {4},
  pages = {045001},
  publisher = {IOP Publishing},
  issn = {2058-9565},
  doi = {10.1088/2058-9565/aa8072},
  langid = {english}
}

@article{collins2022WeingartenCalculusa,
  title = {The {{Weingarten Calculus}}},
  author = {Collins, Benoit and Matsumoto, Sho and Novak, Jonathan},
  year = {2022},
  month = may,
  journal = {Notices of the American Mathematical Society},
  volume = {69},
  number = {05},
  pages = {1},
  issn = {0002-9920, 1088-9477},
  doi = {10.1090/noti2474}
}

@article{martinez2025EfficientSimulationParametrizedQuantumCircuitsNonunitala,
  title = {Efficient {{Simulation}} of {{Parametrized Quantum Circuits}} under {{Nonunital Noise}} through {{Pauli Backpropagation}}},
  author = {Martinez, Victor and Angrisani, Armando and Pankovets, Ekaterina and Fawzi, Omar and Stilck Fran{\c c}a, Daniel},
  year = {2025},
  month = jun,
  journal = {Physical Review Letters},
  volume = {134},
  number = {25},
  pages = {250602},
  publisher = {American Physical Society},
  doi = {10.1103/j1gg-s6zb}
}

@article{uvarov2021BarrenPlateausCostFunctionLocalityVariational,
  title = {On Barren Plateaus and Cost Function Locality in Variational Quantum Algorithms},
  author = {Uvarov, A V and Biamonte, J D},
  year = {2021},
  month = may,
  journal = {Journal of Physics A: Mathematical and Theoretical},
  volume = {54},
  number = {24},
  pages = {245301},
  publisher = {IOP Publishing},
  issn = {1751-8121},
  doi = {10.1088/1751-8121/abfac7},
  langid = {english}
}

@article{patti2021EntanglementDevisedBarrenPlateauMitigation,
  title = {Entanglement Devised Barren Plateau Mitigation},
  author = {Patti, Taylor L. and Najafi, Khadijeh and Gao, Xun and Yelin, Susanne F.},
  year = {2021},
  month = jul,
  journal = {Physical Review Research},
  volume = {3},
  number = {3},
  pages = {033090},
  publisher = {American Physical Society},
  doi = {10.1103/PhysRevResearch.3.033090}
}

@article{kokail2019SelfverifyingVariationalQuantumSimulationLatticeModels,
  title = {Self-Verifying Variational Quantum Simulation of Lattice Models},
  author = {Kokail, C. and Maier, C. and {van Bijnen}, R. and Brydges, T. and Joshi, M. K. and Jurcevic, P. and Muschik, C. A. and Silvi, P. and Blatt, R. and Roos, C. F. and Zoller, P.},
  year = {2019},
  month = may,
  journal = {Nature},
  volume = {569},
  number = {7756},
  pages = {355--360},
  publisher = {Nature Publishing Group},
  issn = {1476-4687},
  doi = {10.1038/s41586-019-1177-4},
  copyright = {2019 The Author(s), under exclusive licence to Springer Nature Limited},
  langid = {english}
}

@article{li2017EfficientVariationalQuantumSimulatorIncorporatingActive,
  title = {Efficient {{Variational Quantum Simulator Incorporating Active Error Minimization}}},
  author = {Li, Ying and Benjamin, Simon C.},
  year = {2017},
  month = jun,
  journal = {Physical Review X},
  volume = {7},
  number = {2},
  pages = {021050},
  publisher = {American Physical Society},
  doi = {10.1103/PhysRevX.7.021050}
}

@article{biamonte2017QuantumMachineLearning,
  title = {Quantum Machine Learning},
  author = {Biamonte, Jacob and Wittek, Peter and Pancotti, Nicola and Rebentrost, Patrick and Wiebe, Nathan and Lloyd, Seth},
  year = {2017},
  month = sep,
  journal = {Nature},
  volume = {549},
  number = {7671},
  pages = {195--202},
  publisher = {Nature Publishing Group},
  issn = {1476-4687},
  doi = {10.1038/nature23474},
  copyright = {2017 Macmillan Publishers Limited, part of Springer Nature. All rights reserved.},
  langid = {english}
}

@article{monz2016RealizationScalableShorAlgorithm,
  title = {Realization of a Scalable {{Shor}} Algorithm},
  author = {Monz, Thomas and Nigg, Daniel and Martinez, Esteban A. and Brandl, Matthias F. and Schindler, Philipp and Rines, Richard and Wang, Shannon X. and Chuang, Isaac L. and Blatt, Rainer},
  year = {2016},
  month = mar,
  journal = {Science},
  volume = {351},
  number = {6277},
  pages = {1068--1070},
  publisher = {American Association for the Advancement of Science},
  doi = {10.1126/science.aad9480}
}

@article{peruzzo2014VariationalEigenvalueSolverPhotonicQuantumProcessor,
  title = {A Variational Eigenvalue Solver on a Photonic Quantum Processor},
  author = {Peruzzo, Alberto and McClean, Jarrod and Shadbolt, Peter and Yung, Man-Hong and Zhou, Xiao-Qi and Love, Peter J. and {Aspuru-Guzik}, Al{\'a}n and O'Brien, Jeremy L.},
  year = {2014},
  month = jul,
  journal = {Nature Communications},
  volume = {5},
  number = {1},
  pages = {4213},
  publisher = {Nature Publishing Group},
  issn = {2041-1723},
  doi = {10.1038/ncomms5213},
  copyright = {2014 The Author(s)},
  langid = {english}
}

@article{mele2024IntroductionHaarMeasureToolsQuantumInformation,
  title = {Introduction to {{Haar Measure Tools}} in {{Quantum Information}}: {{A Beginner}}'s {{Tutorial}}},
  shorttitle = {Introduction to {{Haar Measure Tools}} in {{Quantum Information}}},
  author = {Mele, Antonio Anna},
  year = {2024},
  month = may,
  journal = {Quantum},
  volume = {8},
  pages = {1340},
  issn = {2521-327X},
  doi = {10.22331/q-2024-05-08-1340},
  langid = {english}
}

@article{haug2021CapacityQuantumGeometryParametrizedQuantumCircuits,
  title = {Capacity and {{Quantum Geometry}} of {{Parametrized Quantum Circuits}}},
  author = {Haug, Tobias and Bharti, Kishor and Kim, M.S.},
  year = {2021},
  month = oct,
  journal = {PRX Quantum},
  volume = {2},
  number = {4},
  pages = {040309},
  publisher = {American Physical Society},
  doi = {10.1103/PRXQuantum.2.040309}
}

@article{mehendale2023ExploringParameterRedundancyUnitaryCoupledClusterAnsatze,
  title = {Exploring {{Parameter Redundancy}} in the {{Unitary Coupled-Cluster Ansatze}} for {{Hybrid Variational Quantum Computing}}},
  author = {Mehendale, Shashank G. and Peng, Bo and Govind, Niranjan and Alexeev, Yuri},
  year = {2023},
  month = may,
  journal = {J. Phys. Chem. A},
  volume = {127},
  number = {20},
  pages = {4526--4537},
  issn = {1089-5639, 1520-5215},
  doi = {10.1021/acs.jpca.3c00550}
}

@article{ragone2024LieAlgebraicTheoryBarrenPlateausDeep,
  title = {A {{Lie}} Algebraic Theory of Barren Plateaus for Deep Parameterized Quantum Circuits},
  author = {Ragone, Michael and Bakalov, Bojko N. and Sauvage, Fr{\'e}d{\'e}ric and Kemper, Alexander F. and Ortiz Marrero, Carlos and Larocca, Mart{\'i}n and Cerezo, M.},
  year = {2024},
  month = aug,
  journal = {Nat. Commun.},
  volume = {15},
  number = {1},
  pages = {7172},
  publisher = {Nature Publishing Group},
  issn = {2041-1723},
  doi = {10.1038/s41467-024-49909-3},
  copyright = {2024 The Author(s)},
  langid = {english}
}

@article{arrasmith2021effect,
  title={Effect of barren plateaus on gradient-free optimization},
  author={Arrasmith, Andrew and Cerezo, Marco and Czarnik, Piotr and Cincio, Lukasz and Coles, Patrick J},
  journal={Quantum},
  volume={5},
  pages={558},
  year={2021},
  publisher={Verein zur F{\"o}rderung des Open Access Publizierens in den Quantenwissenschaften}
}

@article{liHybridQuantumClassicalApproachQuantumOptimalControl2017,
  title = {Hybrid {{Quantum-Classical Approach}} to {{Quantum Optimal Control}}},
  author = {Li, Jun and Yang, Xiaodong and Peng, Xinhua and Sun, Chang-Pu},
  year = {2017},
  month = apr,
  journal = {Physical Review Letters},
  volume = {118},
  number = {15},
  primaryclass = {quant-ph},
  pages = {150503},
  issn = {0031-9007, 1079-7114},
  doi = {10.1103/PhysRevLett.118.150503}
}

@article{crooksGradientsParameterizedQuantumGatesUsingParametershift2019,
  title = {Gradients of Parameterized Quantum Gates Using the Parameter-Shift Rule and Gate Decomposition},
  author = {Crooks, Gavin E.},
  year = {2019},
  month = may,
  primaryclass = {quant-ph},
  publisher = {arXiv},
  journal = {arXiv},
  doi = {10.48550/arXiv.1905.13311}
}

@article{schuldQuantumMachineLearningFeatureHilbertSpaces2019,
  title = {Quantum {{Machine Learning}} in {{Feature Hilbert Spaces}}},
  author = {Schuld, Maria and Killoran, Nathan},
  year = {2019},
  month = feb,
  journal = {Physical Review Letters},
  volume = {122},
  number = {4},
  pages = {040504},
  publisher = {American Physical Society},
  doi = {10.1103/PhysRevLett.122.040504}
}

@article{meyerExploitingSymmetryVariationalQuantumMachineLearning2023,
  title = {Exploiting {{Symmetry}} in {{Variational Quantum Machine Learning}}},
  author = {Meyer, Johannes Jakob and Mularski, Marian and {Gil-Fuster}, Elies and Mele, Antonio Anna and Arzani, Francesco and Wilms, Alissa and Eisert, Jens},
  year = {2023},
  month = mar,
  journal = {PRX Quantum},
  volume = {4},
  number = {1},
  pages = {010328},
  publisher = {American Physical Society},
  doi = {10.1103/PRXQuantum.4.010328}
}

@article{endoVariationalQuantumSimulationGeneralProcesses2020,
  title = {Variational {{Quantum Simulation}} of {{General Processes}}},
  author = {Endo, Suguru and Sun, Jinzhao and Li, Ying and Benjamin, Simon C. and Yuan, Xiao},
  year = {2020},
  month = jun,
  journal = {Physical Review Letters},
  volume = {125},
  number = {1},
  pages = {010501},
  publisher = {American Physical Society},
  doi = {10.1103/PhysRevLett.125.010501}
}

@article{shaoSimulatingNoisyVariationalQuantumAlgorithmsPolynomial2024,
  title = {Simulating {{Noisy Variational Quantum Algorithms}}: {{A Polynomial Approach}}},
  shorttitle = {Simulating {{Noisy Variational Quantum Algorithms}}},
  author = {Shao, Yuguo and Wei, Fuchuan and Cheng, Song and Liu, Zhengwei},
  year = {2024},
  month = sep,
  journal = {Physical Review Letters},
  volume = {133},
  number = {12},
  primaryclass = {quant-ph},
  pages = {120603},
  issn = {0031-9007, 1079-7114},
  doi = {10.1103/PhysRevLett.133.120603}
}

@article{shorPolynomialTimeAlgorithmsPrimeFactorizationDiscreteLogarithms1997,
  title = {Polynomial-{{Time Algorithms}} for {{Prime Factorization}} and {{Discrete Logarithms}} on a {{Quantum Computer}}},
  author = {Shor, Peter W.},
  year = {1997},
  month = oct,
  journal = {SIAM Journal on Computing},
  volume = {26},
  number = {5},
  pages = {1484--1509},
  issn = {0097-5397, 1095-7111},
  doi = {10.1137/S0097539795293172}
}

@article{harrowQuantumAlgorithmLinearSystemsEquations2009,
  title = {Quantum {{Algorithm}} for {{Linear Systems}} of {{Equations}}},
  author = {Harrow, Aram W. and Hassidim, Avinatan and Lloyd, Seth},
  year = {2009},
  month = oct,
  journal = {Physical Review Letters},
  volume = {103},
  number = {15},
  pages = {150502},
  publisher = {American Physical Society},
  doi = {10.1103/PhysRevLett.103.150502}
}

@inproceedings{groverFastQuantumMechanicalAlgorithmDatabaseSearch1996,
  title = {A Fast Quantum Mechanical Algorithm for Database Search},
  booktitle = {Proceedings of the Twenty-Eighth Annual {{ACM}} Symposium on {{Theory}} of {{Computing}}},
  author = {Grover, Lov K.},
  year = {1996},
  month = jul,
  series = {{{STOC}} '96},
  pages = {212--219},
  publisher = {Association for Computing Machinery},
  address = {New York, NY, USA},
  doi = {10.1145/237814.237866},
  isbn = {978-0-89791-785-8}
}

@article{kandalaHardwareefficientVariationalQuantumEigensolverSmallMolecules2017,
  title = {Hardware-Efficient Variational Quantum Eigensolver for Small Molecules and Quantum Magnets},
  author = {Kandala, Abhinav and Mezzacapo, Antonio and Temme, Kristan and Takita, Maika and Brink, Markus and Chow, Jerry M. and Gambetta, Jay M.},
  year = {2017},
  month = sep,
  journal = {Nature},
  volume = {549},
  number = {7671},
  pages = {242--246},
  publisher = {Nature Publishing Group},
  issn = {1476-4687},
  doi = {10.1038/nature23879},
  copyright = {2017 Macmillan Publishers Limited, part of Springer Nature. All rights reserved.}
}

@article{carolanVariationalQuantumUnsamplingQuantumPhotonicProcessor2020,
  title = {Variational Quantum Unsampling on a Quantum Photonic Processor},
  author = {Carolan, Jacques and Mohseni, Masoud and Olson, Jonathan P. and Prabhu, Mihika and Chen, Changchen and Bunandar, Darius and Niu, Murphy Yuezhen and Harris, Nicholas C. and Wong, Franco N. C. and Hochberg, Michael and Lloyd, Seth and Englund, Dirk},
  year = {2020},
  month = mar,
  journal = {Nature Physics},
  volume = {16},
  number = {3},
  pages = {322--327},
  publisher = {Nature Publishing Group},
  issn = {1745-2481},
  doi = {10.1038/s41567-019-0747-6},
  copyright = {2020 The Author(s), under exclusive licence to Springer Nature Limited}
}

@article{angrisaniClassicallyEstimatingObservablesNoiselessQuantumCircuits,
  title = {Classically Estimating Observables of Noiseless Quantum Circuits},
  author = {Angrisani, Armando and Schmidhuber, Alexander and Rudolph, Manuel S and Cerezo, M and Holmes, Zoe and Huang, Hsin-Yuan},
  journal = {Physical Review Letters},
  year = {2025},
  month = {Sep},
  publisher = {American Physical Society},
  doi = {10.1103/lh6x-7rc3},
}

@article{cerezoCostFunctionDependentBarrenPlateausShallow2021,
  title = {Cost Function Dependent Barren Plateaus in Shallow Parametrized Quantum Circuits},
  author = {Cerezo, M. and Sone, Akira and Volkoff, Tyler and Cincio, Lukasz and Coles, Patrick J.},
  year = {2021},
  month = mar,
  journal = {Nat Commun},
  volume = {12},
  number = {1},
  pages = {1791},
  publisher = {Nature Publishing Group},
  issn = {2041-1723},
  doi = {10.1038/s41467-021-21728-w},
  copyright = {2021 The Author(s)}
}

@article{cerezo2025DoesProvableAbsenceBarrenPlateausImply,
  title = {Does Provable Absence of Barren Plateaus Imply Classical Simulability?},
  author = {Cerezo, M. and Larocca, Martin and {Garc{\'i}a-Mart{\'i}n}, Diego and Diaz, N. L. and Braccia, Paolo and Fontana, Enrico and Rudolph, Manuel S. and Bermejo, Pablo and Ijaz, Aroosa and Thanasilp, Supanut and Anschuetz, Eric R. and Holmes, Zo{\"e}},
  year = {2025},
  month = aug,
  journal = {Nature Communications},
  volume = {16},
  number = {1},
  pages = {7907},
  issn = {2041-1723},
  doi = {10.1038/s41467-025-63099-6}
}

@article{laroccaDiagnosingBarrenPlateausToolsQuantumOptimal2022,
  title = {Diagnosing {{Barren Plateaus}} with {{Tools}} from {{Quantum Optimal Control}}},
  author = {Larocca, Martin and Czarnik, Piotr and Sharma, Kunal and Muraleedharan, Gopikrishnan and Coles, Patrick J. and Cerezo, M.},
  year = {2022},
  month = sep,
  journal = {Quantum},
  volume = {6},
  pages = {824},
  publisher = {Verein zur F{\"o}rderung des Open Access Publizierens in den Quantenwissenschaften},
  doi = {10.22331/q-2022-09-29-824}
}

@article{mccleanBarrenPlateausQuantumNeuralNetworkTraining2018,
  title = {Barren Plateaus in Quantum Neural Network Training Landscapes},
  author = {McClean, Jarrod R. and Boixo, Sergio and Smelyanskiy, Vadim N. and Babbush, Ryan and Neven, Hartmut},
  year = {2018},
  month = nov,
  journal = {Nat Commun},
  volume = {9},
  number = {1},
  pages = {4812},
  publisher = {Nature Publishing Group},
  issn = {2041-1723},
  doi = {10.1038/s41467-018-07090-4},
  copyright = {2018 The Author(s)}
}

@article{ortizmarreroEntanglementInducedBarrenPlateaus2021,
  title = {Entanglement-{{Induced Barren Plateaus}}},
  author = {Ortiz Marrero, Carlos and Kieferov{\'a}, M{\'a}ria and Wiebe, Nathan},
  year = {2021},
  month = oct,
  journal = {PRX Quantum},
  volume = {2},
  number = {4},
  pages = {040316},
  publisher = {American Physical Society},
  doi = {10.1103/PRXQuantum.2.040316}
}

@article{pesahAbsenceBarrenPlateausQuantumConvolutionalNeural2021,
  title = {Absence of {{Barren Plateaus}} in {{Quantum Convolutional Neural Networks}}},
  author = {Pesah, Arthur and Cerezo, M. and Wang, Samson and Volkoff, Tyler and Sornborger, Andrew T. and Coles, Patrick J.},
  year = {2021},
  month = oct,
  journal = {Physical Review X},
  volume = {11},
  number = {4},
  pages = {041011},
  publisher = {American Physical Society},
  doi = {10.1103/PhysRevX.11.041011}
}

@misc{ragoneUnifiedTheoryBarrenPlateausDeepParametrized2023,
  title = {A {{Unified Theory}} of {{Barren Plateaus}} for {{Deep Parametrized Quantum Circuits}}},
  author = {Ragone, Michael and Bakalov, Bojko N. and Sauvage, Fr{\'e}d{\'e}ric and Kemper, Alexander F. and Marrero, Carlos Ortiz and Larocca, Martin and Cerezo, M.},
  year = {2023},
  month = sep,
  number = {arXiv:2309.09342},
  primaryclass = {quant-ph},
  publisher = {arXiv}
}

@article{sackAvoidingBarrenPlateausUsingClassicalShadows2022,
  title = {Avoiding {{Barren Plateaus Using Classical Shadows}}},
  author = {Sack, Stefan H. and Medina, Raimel A. and Michailidis, Alexios A. and Kueng, Richard and Serbyn, Maksym},
  year = {2022},
  month = jun,
  journal = {PRX Quantum},
  volume = {3},
  number = {2},
  pages = {020365},
  publisher = {American Physical Society},
  doi = {10.1103/PRXQuantum.3.020365}
}

@article{uvarovBarrenPlateausCostFunctionLocalityVariational2021,
  title = {On Barren Plateaus and Cost Function Locality in Variational Quantum Algorithms},
  author = {Uvarov, A V and Biamonte, J D},
  year = {2021},
  month = may,
  journal = {J. Phys. A: Math. Theor.},
  volume = {54},
  number = {24},
  pages = {245301},
  publisher = {IOP Publishing},
  issn = {1751-8121},
  doi = {10.1088/1751-8121/abfac7}
}

@article{wangNoiseinducedBarrenPlateausVariationalQuantumAlgorithms2021,
  title = {Noise-Induced Barren Plateaus in Variational Quantum Algorithms},
  author = {Wang, Samson and Fontana, Enrico and Cerezo, M. and Sharma, Kunal and Sone, Akira and Cincio, Lukasz and Coles, Patrick J.},
  year = {2021},
  month = nov,
  journal = {Nat Commun},
  volume = {12},
  number = {1},
  pages = {6961},
  publisher = {Nature Publishing Group},
  issn = {2041-1723},
  doi = {10.1038/s41467-021-27045-6},
  copyright = {2021 The Author(s)}
}

@article{xuScramblingDynamicsOutofTimeOrderedCorrelatorsQuantumManyBody2024,
  title = {Scrambling {{Dynamics}} and {{Out-of-Time-Ordered Correlators}} in {{Quantum Many-Body Systems}}},
  author = {Xu, Shenglong and Swingle, Brian},
  year = {2024},
  month = jan,
  journal = {PRX Quantum},
  volume = {5},
  number = {1},
  pages = {010201},
  issn = {2691-3399},
  doi = {10.1103/PRXQuantum.5.010201}
}
\bibliographystyle{apsrev4-2}

\newpage
\appendix
\clearpage
\onecolumngrid

\end{document}